\documentclass[a4paper,12pt,justification=centering]{article}
\usepackage{latexsym}
\usepackage{rotating}
\usepackage{amsmath,amssymb,amsxtra,amscd,amsthm}
\usepackage[mathscr]{eucal}
\usepackage{graphicx}
\usepackage[justification=centering]{caption}
\usepackage{subfig}
\usepackage{datetime}
\usepackage{pdfsync}
\usepackage{thumbpdf}
\usepackage{color}
\ifx\pdfoutput\undefined
\usepackage
[dvips,
 verbose=false,
 bookmarks=true,
 colorlinks=true,
 linkcolor=webred,
 filecolor=webbrown,
 citecolor=webgreen,
 pagecolor=webblue,
 urlcolor=webblue,
 pdftitle={},
 pdfauthor={Murad Alim},
 pdfsubject={},
 pdfkeywords={},
 bookmarksopen=false,
 pdfpagemode=None,
 pdfview=FitH,
 pdfstartview=FitH,
 extension=pdf]{hyperref}
\else
\usepackage
[pdftex,
 verbose=false,
 bookmarks=true,
 colorlinks=true,
 linkcolor=webred,
 filecolor=webbrown,
 citecolor=webgreen,
 pagecolor=webblue,
 urlcolor=webblue,
 pdftitle={},
 pdfauthor={Murad Alim},
 pdfsubject={},
 pdfkeywords={},
 bookmarksopen=false,
 pdfpagemode=None,
 pdfview=FitH,
 pdfstartview=FitH,
 extension=pdf]{hyperref}
\fi
\definecolor{webred}{rgb}{.8,0,0}
\definecolor{webbrown}{rgb}{.6,0,0}
\definecolor{webgreen}{rgb}{0,0.5,0}
\definecolor{webdkgreen}{rgb}{0,0.3,0}
\definecolor{webblue}{rgb}{0,0,0.5}

\numberwithin{equation}{section}

\providecommand{\href}[2]{#2}

\allowdisplaybreaks
\usepackage{bbm}
\usepackage{rotating}
\newcommand{\be}{\begin{eqnarray}}
\newcommand{\beq}{\begin{eqnarray}}
\newcommand{\ee}{\end{eqnarray}}
\newcommand{\ov}{\overline}

\newcommand{\ib}{\ov{\imath}}
\newcommand{\jb}{\ov{\jmath}}
\newcommand{\kb}{\ov{k}}

\newcommand{\CY}{\mathcal{X}}

\def\e#1\e{\begin{equation}#1\end{equation}}
\def\ea#1\ea{\begin{align}#1\end{align}}

\theoremstyle{plain}
\newtheorem{thm}{Theorem}[section]
\newtheorem*{thm*}{Theorem}

\newtheorem{prop}[thm]{Proposition}

\theoremstyle{definition}
\newtheorem{dfn}[thm]{Definition}

\newtheorem{rem}[thm]{Remark}

\def\tgtwo{\mathfrak{t}}
\def\tgone{\mathfrak{t}}
\def\tgzero{\mathfrak{t}}
\def\kgone{\mathfrak{l}}
\def\ggtwo{\mathfrak{g}}
\def\ggzero{\mathfrak{g}}

\def\Z{\mathbb{Z}}                   
\def\tr{{{\mathsf t}{\mathsf r}}}                 
\def\T{{\sf T}}                      

\def\BG{{\sf G}}                       
\def\gm{{\sf  A}}                    

\def\Yuk{C}                 
\def\Ra{{\sf R}}                      
\def\g{{\sf g}}                       
\def\t{{\sf t}}                       

\def\dim{{\rm    dim}}                
\def\Lie{{\rm Lie}}                   
\def\gg{{\mathfrak g}}                

\def\imc{{  \Phi }}                  
\def\imcd{{ \Phi_{CY_d}}}

\def\GL{{\rm GL}}                

\def\Mat{{\rm Mat}}              

\begin{document}

\setlength{\parindent}{0cm}
\setlength{\baselineskip}{1.5em}
\title{Algebraic structure of $tt^*$ equations \\ for Calabi-Yau sigma models}
\author{Murad Alim\footnote{\tt{alim@physics.harvard.edu}}
\\
\small Department of Mathematics, Harvard University,\\ \small 1 Oxford Street, Cambridge, MA 02138, USA\\
\small Jefferson Physical Laboratory, Harvard University, \\ \small 17 Oxford Street, Cambridge, MA 02138, USA\\
}

\date{}
\maketitle

\abstract{The $tt^*$ equations define a flat connection on the moduli spaces of $2d, \mathcal{N}=2$ quantum field theories. For conformal theories with $c=3d$, which can be realized as nonlinear sigma models into Calabi-Yau d-folds, this flat connection is equivalent to special geometry for threefolds and to its analogs in other dimensions. We show that the non-holomorphic content of the $tt^*$ equations in the cases $d=1,2,3$ is captured in terms of finitely many generators of special functions, which close under derivatives. The generators are understood as coordinates on a larger moduli space. This space parameterizes a freedom in choosing representatives of the chiral ring while preserving a constant topological metric. Geometrically, the freedom corresponds to a choice of forms on the target space respecting the Hodge filtration and having a constant pairing. Linear combinations of vector fields on that space are identified with generators of a Lie algebra. This Lie algebra replaces the non-holomorphic derivatives of $tt^*$ and provides these with a finer and algebraic meaning. For sigma models into lattice polarized $K3$ manifolds, the differential ring of special functions on the moduli space is constructed, extending known structures for $d=1$ and 3. The generators of the differential rings of special functions are given by quasi-modular forms for $d=1$ and their generalizations in $d=2,3$. Some explicit examples are worked out including the case of the mirror of the quartic in $\mathbbm{P}^3$, where due to further algebraic constraints, the differential ring coincides with quasi modular forms.}

\clearpage


\tableofcontents


\section{Introduction}

Two dimensional, $\mathcal{N}=(2,2)$, superconformal field theories (SCFT) are at the heart of many exciting developments and interactions between mathematics and physics. These theories have a finite set of special states which exhibit a ring structure, the chiral ring \cite{Lerche:1989uy}. The realization of these theories as non-linear sigma models (NLSM) into Calabi-Yau (CY) threefold target spaces and two twists lead to two topological field theories, the A- and the B-model \cite{Witten:1988xj,Witten:1991zz}, which are exchanged by mirror symmetry.

The deformations of the SCFT are generated by the marginal fields, which are also in the chiral ring. The techniques of mirror symmetry rely on recovering in a geometric context the chiral ring structure. In geometric realizations of the superconformal algebra (SCA) this becomes a problem of variation of Hodge structure on the middle dimensional cohomology on the B-side, with its flat Gauss-Manin connection. On the A-side the variation problem introduces the notion of quantum cohomology and can be phrased in terms of a variation of Hodge structure with a flat connection, the A-model connection. Mirror symmetry is established by matching the two. See e.g. Refs.~\cite{Cox:1999,Aspinwall:book,Alim:2012gq} for more details.
 
Building on the chiral ring of $2d$ SCFT and similar structures for Landau-Ginzburg models developed in Refs.~\cite{Cecotti:1990wz,Cecotti:1990kz}, Cecotti and Vafa introduced in Ref.~\cite{Cecotti:1991me} the $tt^*$ equations which describe the variation of the ground states of $2d, \,\mathcal{N}=2$ quantum field theories, which are not necessarily conformal. For the conformal cases with central charge $c=3d$, which can be realized as nonlinear sigma models into CY d-folds the $tt^*$ equations are equivalent to special geometry of CY threefolds \cite{Strominger:1990pd} and to ita analogs in other dimensions. In these cases the chiral ring translates into the cohomology of the target CY. When a subset of the chiral ring is considered, the $tt^*$ equations describe the variation of Hodge structure. 

This variation of Hodge structure is equipped with a natural holomorphic flat connection, the Gauss-Manin connection. The $tt^*$ equations on the other hand also define a flat connection on the moduli space, which has both a holomorphic and a non-holomorphic component. This connection is sometimes called the $tt^*$ Gauss-Manin or the $tt^*$ Lax connection, see Ref.~\cite{Cecotti:2013mba} for a recent treatment of this. The question which this work is addressing and answering is how the non-holomorphic content of the $tt^*$ connection for CY sigma models can be phrased in a purely holomorphic and even algebraic way. 

The formulation of the answer builds on previous mathematical work, in particular on an extension of the Gauss-Manin connection to moduli spaces of Calabi-Yau spaces enhanced with choices of differential forms developed by Movasati. In Ref.~\cite{Mov12c}, the enhanced moduli space of an elliptic curve was considered. This gives a three dimensional moduli space instead of the usual one dimensional space. Moreover, coordinates on this space correspond to quasi-modular forms and special vector fields, when composed with the Gauss-Manin connection, provide an $sl(2,\mathbbm{C})$ Lie algebra. This approach was extended to the mirror quintic CY threefold \cite{Mov11,Mov12a}, see also Ref.~\cite{Mov13} for an overview of this program.

Another strand of works which feeds into the current investigation is the development of polynomial differential rings on the moduli spaces of CY threefolds. It was shown by Yamaguchi and Yau in Ref.~\cite{Yamaguchi:2004bt}, that the non-holomorphic content of the special geometry of the mirror quintic threefold is entirely captured in terms of finitely many functions which close under derivatives. This was generalized in Ref.~\cite{Alim:2007qj} for arbitrary CY threefolds and further studied in Ref.~\cite{Hosono:2008np}, where this structure was also developed for the elliptic curve. These differential rings coincide with those of Ref.~\cite{Mov11} and are expected to provide generalizations of classical quasi modular forms. When subgroups of $SL(2,\mathbbm{Z})$ appear in limits of the monodromy groups of the mirror CY threefolds \cite{Alim:2012ss,Alim:2013eja}, the differential rings coincide with the differential rings of quasi modular forms of Kaneko and Zagier \cite{Kaneko:1995}.

A synthesis of these works for CY threefolds was performed in Ref.~\cite{Alim:2014dea}. In that work the enhanced moduli space of CY threefolds was studied and it was shown that the polynomial generators of Ref.~\cite{Alim:2007qj} provide a parameterization thereof. Moreover a Lie algebra structure was put forward which allows an algebraic reformulation of the holomorphic anomaly equations of Bershadsky, Cecotti, Ooguri and Vafa (BCOV) \cite{Bershadsky:1993ta,Bershadsky:1993cx}. 

The purpose of the present work is to tie different strands together to provide an algebraic formulation of the $tt^*$ equations for the cases of nonlinear sigma models into CY targets of dimensions $1,2$ and $3$. In order to achieve this, new structures will be developed including an extension of the differential rings of special functions to the moduli spaces of lattice polarized $K3$ manifolds, as well as explicit expressions for the constraints on the K\"ahler metric which are imposed by the flatness of the $tt^*$ connection.
These constraints are strong enough to reproduce the solution of the $tt^*$ equations for the elliptic curve obtained in Refs.~\cite{Cecotti:1990kz,Cecotti:1991me} by different means as well as to put forward the analog for lattice polarized $K3$ manifolds.

For the nonlinear sigma models into CY manifolds of dim $d=1,2,3$ which are considered here, it will furthermore be shown, that the generators of the special differential rings serve as parameters on a larger moduli space $\mathsf{T}$. This space corresponds to the moduli space of the CY enhanced with a choice of differential forms which respect the Hodge filtration and have a constant pairing, which is symplectic ($d=1,3$) or symmetric $(d=2)$. From the $2d$ point of view this moduli space parameterizes a freedom in choosing representatives of the chiral ring elements which lie in the deformation bundle. The latter consists of all the states which are obtained from the repeated action on the Neveu-Schwarz (NS) ground state by the marginal deformations; via spectral flow this corresponds to a subset of the Ramond (R) ground states of the theory. 

Considering linear combinations of vector fields along the generators on $\mathsf{T}$ a Lie algebra structure will be put forward following Refs.~\cite{Mov12c,Alim:2014dea}. This Lie algebra is a finer and algebraic equivalent of the $tt^*$ equations in these cases. In examples it will be furthermore shown, that coordinates on this larger space $\mathsf{T}$ correspond to quasi modular forms and their generalizations.

While the present text only treats $tt^*$ equations of CY sigma models of dimensions $1,2,3$, a generalization to higher dimensional CY target spaces, using the results of Ref.~\cite{Greene:1993vm} is expected and a generalization to LG models and general $tt^*$ geometries is conceivable. A holomorphic and algebraic reformulation of the $tt^*$ equations is relevant for a number of reasons. Physically, the $tt^*$ equations and the mere fact that solutions can be found are both mysterious and miraculous \cite{Cecotti:1991me,Cecotti:1991vb,Cecotti:1992vy}. In fact, for finding solutions in Refs.~\cite{Cecotti:1990kz,Cecotti:1991me} the symmetry content of a Hodge theoretic approach was used to identify the problem given by the $tt^*$ equations in various examples with classical equations of integrable systems, such as Toda equations \cite{Cecotti:1991me,Cecotti:1991vb,Cecotti:1992vy}. A clear reformulation of the equations in manifestly algebraic terms should push further connections to classical problems of mathematical physics. Furthermore, non-compact CY threefolds are used to geometrically engineer $4d, \mathcal{N}=2$ gauge theories \cite{Klemm:1996bj}, the CY moduli spaces are identified with the Coulomb branch of these theories. In this context the non-holomorphic content of the $tt^*$ equations is only relevant when corrections are considered which correspond to Nekrasov's deformations \cite{Nekrasov:2002qd}, which can be captured by higher genus topological string amplitudes; see for example Ref.~\cite{Huang:2006si}. The present work suggests that it is also meaningful to consider in a holomorphic language the content of the non-holomorphic deformations of $tt^*$ for the field theories. The larger moduli space advocated in this work, corresponding to the choices of chiral ring elements in the $2d$ world, translates to a parameterization of a choice of two differentials, one of which is the Seiberg-Witten differential \cite{Seiberg:1994rs}, which contains the $4d$ physics. 

On the mathematical side, phrasing the non-holomorphic content of the $tt^*$ equations in a holomorphic way may simplify a rigorous treatment of these equations; see e.~g.~Ref.~\cite{Hertling:review} and references therein. Furthermore for the CY sigma models considered in the present work, the algebraic treatment suggests an extension of the classical study of the variation of Hodge structures, which connects it beautifully to the world of modular forms. This point of view has been advocated in Refs.~\cite{Mov11,Mov12a,Mov12b,Mov12c,Mov13} without the use of $tt^*$ geometry.


\section{Summary and Outline}
In this work we consider the realization of the $\mathcal{N}=(2,2)$ superconformal algebra with $(c,\bar{c})=(3d,3d)$ as nonlinear sigma models into families of CY $d$-fold mirror target spaces $\check{\CY}_t,\CY_z$. We consider the topological twists leading to the A-model on the family $\check{\CY}_t$ and to the B-model on the family $\CY_z$. The moduli spaces $\mathcal{M}$ of the target spaces will correspond to the complexified K\"ahler moduli space of $\check{\CY}_t$ and to the complex structure moduli space of $\CY_z$. We will consider the cases $d=1,2,3$ corresponding to the elliptic curve, lattice polarized K3 manifolds and Calabi-Yau threefolds respectively.

The $tt^*$ equations stem from the $2d$ perspective and are valid on both sides of mirror symmetry but it is much more challenging to provide their geometric meaning on the A-model side. Likewise, in this work the algebraic structure of the $tt^*$ equations will be developed in the B-model context so we will focus on the families of CY d-folds $\CY_z$ and study the variation of Hodge structure of the middle dimensional cohomology $H^d(\CY_z,\mathbbm{C})$ over the moduli space of complex structure $\mathcal{M}$. By mirror symmetry, the same structures should be recovered on the A-model side as well.

\subsection{Larger moduli space $\mathsf{T}$}
The main ingredient in the analysis, is to consider a larger moduli space $\mathsf{T}$, following Refs.~\cite{Mov12c,Alim:2014dea}, which parameterizes pairs:
\begin{equation}
\left( \CY_z\, , \vec{\Omega} \right)\,,
\end{equation}
where $\vec{\Omega}= \left( \omega_1\,\dots \omega_N \right)$ denotes a choice of $N$ elements in $H^d(\CY,\mathbbm{C})$ which respect the Hodge filtration:
\begin{equation}
0= F^{d+1} \subset F^{d} \subset \dots \subset F^0 = H^{d}_{dR}(\CY_z)\,,
\end{equation}
this means $\omega_1\in F^d, (\omega_1,\omega_2,\dots, \omega_{h^{d-1,1}+1}) \in F^{d-1}\,,\dots,(\omega_1,\dots,\omega_N) \in F^{0}$. $N$ is given by $2, 2+ n, 2n+2$ for $d=1,2,3$ respectively, where $n=\dim \mathcal{M}$. For lattice polarized K3 manifolds $n$ is the rank of the Picard lattice of $\check{\CY}$, for threefolds $n=h^{2,1}(\CY)$.\footnote{The case of $d=2$ requires a more careful analysis, in the general discussion of $H^{d}(\CY,\mathbbm{C})$ we mean the part coming form the transcendental lattice of $\CY$, this will be discussed in Sec.~\ref{twofolds}} We further impose on the choice of elements in the filtration spaces, the following condition:
\begin{equation}
\langle \vec{\Omega},\vec{\Omega} \rangle = \imcd\,,
\end{equation}
where the pairing $\langle .,.\rangle$ is given by the cup product and where $\imcd$ is given by:
\begin{equation}
\Phi_{CY_1}:=\left(
\begin{array}{cc}
0&1\\
-1& 0
\end{array}\right)
\,,\quad 
\Phi_{CY_2}:=
\left(\begin{array}{ccc}
0&0&-1\\
0&C_{ab}&0\\
-1& 0&0
\end{array}\right)
\,,\quad
\Phi_{CY_3}:=
\left(\begin{array}{cccc}
0&0&0&1\\
0& 0&\mathbbm{1}_{n \times n}&0\\
0& -\mathbbm{1}_{n\times n}&0&0\\
 -1&0&0&0
\end{array}\right)\,,
\end{equation}
where $C_{ab}$ is an $n\times n$ matrix giving the intersection form on the Picard lattice of $\check{\CY}$. 
In the SCFT language $n$ denotes the number of marginal deformations and the choice of forms corresponds to choices of representatives of the chiral ring elements lying in the deformation bundle, which will be introduced in the main body of this work and which corresponds to all the elements obtained from the repeated action of the marginal operators on the unique NS ground state.

The change of choices of elements in $\vec{\Omega}$ is captured by the algebraic group:
\begin{equation}
\BG:=\left \{\g\in \GL(N,\mathbbm{C})\mid \g \text{ is block upper triangular and  }  \g^\tr\imcd\g=\imcd\ \ \right \}
\end{equation}
acts from the right on $\T$ and its Lie algebra is given by: 
\begin{equation}
\Lie(\BG)=
\left \{\gg\in \Mat(N,\mathbbm{C})\mid \gg \text{ is block upper triangular and  }  \gg^\tr\imcd+\imcd\gg=0\ \ \right \}.
\end{equation}

The $tt^*$ geometry describes the variation of the Hodge decomposition for the CY d-folds. Starting from the holomorphic $(d,0)$ form $\Omega$, one can define:
\begin{equation}
e^{-K}:= c_d \int_\CY \Omega \wedge \overline{\Omega}\, \in \Gamma(\mathcal{L} \otimes \ov{\mathcal{L}})\,,
\end{equation}
where $c_d$ is a normalization factor which depends on the conventions. $K$ defines a K\"ahler potential and a K\"ahler metric is given by :
\begin{equation}
G_{i\ib} :=\partial_i \partial_{\ib} K\,,
\end{equation}
where $\partial_i=\frac{\partial}{\partial z^i}$ with respect to some local coordinates $z^i,i=1,\dots,n$ on $\mathcal{M}$. The Levi-Civita connection is then given by:
\begin{equation}\label{chris}
\Gamma_{ij}^k = G^{k\kb} \partial_i G_{j\kb}\,.
\end{equation}
The Hodge decomposition is given by:
\begin{equation} 
 H^d(\CY,\mathbbm{C}) \simeq \bigoplus_{p+q=d} H^{p,q}(\CY)\, .
\end{equation}
A choice of elements in the decomposition spaces is given by: 
\begin{equation}
\vec{\Omega}_{nh}= \left( \Omega\, \quad D_i \Omega\, \dots   \ov{\Omega}\right)\,,
\end{equation}
where all the elements are obtained by taking repeated covariant derivatives of the holomorphic top form, $D_i\Omega \in H^{d-1,1}(\CY)$ and in general there are relations between the multi-derivatives generating the next elements. It will be shown that this choice of Hodge decomposition satisfies the $tt^*$ equations:
 \begin{eqnarray}
\left[ D_i,\mathcal{C}_{\ib}\right]&=& 0\,, \\
\left[ D_{\ib},\mathcal{C}_{i}\right]&=&0\,, \\
\left[ D_i,D_{\ib}\right] &=& - \left[ \mathcal{C}_{i},\mathcal{C}_{\ib}\right]\,, \\
\left[ D_i,\mathcal{C}_j\right] &=& \left[ D_j,\mathcal{C}_i\right]\,,\\
\left[ D_{\ib},\mathcal{C}_{\jb}\right] &=& \left[ D_{\jb},\mathcal{C}_{\ib}\right]\,,
\end{eqnarray}
where $\mathcal{C}_i,(\mathcal{C}_{\ib})$ denote $N\times N$ matrices with strictly upper (lower) triangular entries. This gives the $tt^*$ connection:
\begin{equation}
\nabla_i= D_i + \frac{1}{\zeta} \, \mathcal{C}_i\,, \quad  \nabla_{\ib}= D_{\ib} + \zeta\, \mathcal{C}_{\ib}\,,
\end{equation}
which is flat:
\begin{equation}
\left[ \nabla_i,\nabla_{\ib}\right]=0\,,\quad \left[ \nabla_i,\nabla_j\right]=0\,,\quad \left[ \nabla_{\ib},\nabla_{\jb}\right]=0\,.
\end{equation}

\subsection{Algebraic description}
The essence of this work is to show that the holomorphic and non-holomorphic components of the $tt^*$ connection get replaced by:
\begin{equation}
\nabla_i \rightarrow \nabla_{\Ra_a} \quad \nabla_{\ib} \rightarrow \ \nabla_{\Ra_{\gg}} \quad \gg\in \Lie(\BG)\,.
\end{equation}
where $\nabla_{\Ra_a},\nabla_{\Ra_{\gg}}$ denote the composition of vector fields $\Ra$ in $\mathsf{T}$ with the holomorphic Gauss-Manin connection. 
In order to show this, a holomorphic choice of filtration elements $\vec{\Omega}_z$ using multi-derivatives with respect to arbitrary local coordinates $z^i$ on the moduli space $\mathcal{M}$ is related to a distinguished choice $\vec{\Omega}_t$, which uses a special normalization of $\Omega\in H^{d,0}(\CY)$ as well as distinguished \emph{special} coordinates to generate the other elements. The distinguished choice has furthermore the property that the matrices $\mathcal{C}_a$ become strictly upper triangular. The crucial point is that the distinguished holomorphic choice can be obtained by using the $tt^*$ special geometry as intermediate steps to generate the elements of the Hodge decomposition described above. In doing so, multi-derivatives of $tt^*$ metrics and connections will appear. It is shown for all the cases considered in this work that these multi-derivatives of the geometric objects form a closed ring in finitely many generators, these generators correspond to non-holomorphic potentials for the entries of $\mathcal{C}_{\ib}$. In relating the arbitrary choice to the distinguished choice thus the generators of these differential rings show up. By considering the holomorphic limits of these objects one obtains a purely holomorphic parameterization of the freedom which one has in choosing the elements in the filtration spaces.

\subsection{Main results}

We consider $\vec{\Omega}_t$ as living in the larger space $\mathsf{T}$, which parameterizes different choices of entries of $\vec{\Omega}$. Considering vector fields $\Ra$ in $\mathsf{T}$ we can compose these with the Gauss-Manin connection and define the following:
\begin{equation}
\nabla_{\Ra} \vec{\Omega} = \gm_{\Ra} \vec{\Omega}\,,
\end{equation}
where $\gm_{\Ra}$ are given by $N\times N$ matrices.

\begin{thm}\label{vectorfields}
\begin{itemize}
\item
There are unique vector fields $\Ra_a, \ a=1,2,\ldots,n=\dim \mathcal{M}$ in $\T$ and unique $N\times N$ matrices $\mathcal{C}_a$, such that $\mathcal{C}_a=\gm_{\Ra_a}$ are given by:
\begin{eqnarray}
\mathcal{C}_a &=& \left(\begin{array}{cc} 0 & 1 \\ 0 &0 \end{array}\right)\,,\quad \textrm{for d=1}\,,\\
\mathcal{C}_a &=& \left(\begin{array}{ccc} 0 & \delta_a^j & 0 \\ 0 &0& C_{ai} \\0&0&0 \end{array} \right)\,,\quad \textrm{for d=2} \,,\\
\mathcal{C}_a &=& \left(\begin{array}{cccc} 0 & \delta_a^j & 0&0 \\ 0 &0& C_{aij}&0 \\0&0&0&-\delta_a^i\\0&0&0&0\end{array} \right)\,,\quad \textrm{for d=3} \,,
\end{eqnarray}
 and furthremore:
\begin{equation}
\Ra_a \mathcal{C}_b= \Ra_b \mathcal{C}_a\,.
\end{equation}

\item For any $\gg\in \Lie(\BG)$ there is also a unique vector field $\Ra_{\gg}$  in $\T$ 
such that
\begin{equation}
\gm_{\Ra_\gg}=\gg^{\tr}.
\end{equation}
\end{itemize}
\end{thm}
This theorem is proved case by case in Sec.~\ref{liealgebra}. It was proved for CY threefolds in Ref.~\cite{Alim:2014dea}, the analogous statement for the elliptic curve developed here recovers a result of Movasati \cite{Mov12c}, we also prove by construction that this theorem is true for lattice polarized $K3$ manifolds. 

While developing the necessary structure to prove this theorem, explicit constraints on the curvature of the K\"ahler metric $G_{i\ib}$ on the moduli space $\mathcal{M}$ are spelled out, we obtain the following:

\begin{thm}\label{curvaturethm}
\begin{itemize}
\item The curvature of the K\"ahler metric on $\mathcal{M}$ of elliptic curves can be expressed as:
\begin{equation}
\partial_{\bar{z}} \Gamma_{zz}^z= G_{z\bar{z}}+ e^{2K} C_{z} \ov{C}_{\bar{z}}\,.
\end{equation}

\item For CY twofolds we obtain:
\begin{equation}
\partial_{\bar{m}} \Gamma_{im}^k=\delta^{k}_i G_{m\bar{m}} +\delta_m^k G_{i\bar{m}} + e^{2K} C_{im} \, \ov{C}_{\bar{m}\bar{k}} G^{k\kb}\,.
\end{equation}
Furthermore, the metric satisfies:
\begin{equation}
 G_{m\bar{m}} = - C_{mj} \,\ov{C}_{\bar{m}\jb}\, G^{j\jb}\, e^{2K}\,.
 \end{equation}
 
 \item For CY  threefolds we recover the well known; see e.~g.~Refs.~\cite{Strominger:1990pd,Bershadsky:1993cx}:
\begin{equation}
\partial_{\bar{m}} \Gamma_{im}^k=\delta^{k}_i G_{m\bar{m}} +\delta_m^k G_{i\bar{m}} - e^{2K} C_{ijm} \, \ov{C}_{\bar{m}\bar{k}\jb} G^{j\jb}G^{k\kb}\,.
\end{equation}
 \end{itemize}
\end{thm}

This theorem is proved case by case in Sec.~\ref{flatconnections}.

Using distinguished special coordinates on the moduli space, we can furthermore give explicit solutions of the $tt^*$ metrics, which are given by the following:
\begin{thm}\label{metricthm}
\begin{itemize}
\item The K\"ahler potential $K$ and metric (solution of $tt^*$ equations) takes the form \cite{Cecotti:1990wz}:
\begin{equation}
K =- \log\left(2 |\pi^0|^2 \, \Im \t\,\right), \quad G_{z\bar{z}}= \frac{1}{4|\Im\, t|^2} \left| \frac{\partial t}{\partial z}\right|^2\,,
\end{equation}
where $\Im t= \frac{1}{2i} (t-\bar{t}).$
\item  The K\"ahler potential $K$ and K\"ahler metric $G_{i\jb}$ on $\mathcal{M}$ (solution to the $tt^*$ equations) for CY twofolds is given by:
\begin{eqnarray}
K &=& -\log \left( 2 |\pi^0|^2 C_{ab} \Im t^a \Im t^b \right)\,,\\
G_{i\jb} &=&\frac{1}{2  (C_{cd} \Im\,t^c \Im\,t^d)^2} \frac{\partial t^a}{\partial z^i} \frac{\partial \ov{t}^{b}}{\partial z^{\jb}} \left( 2 C_{ad} C_{bc} \Im\,t^c \Im\,t^d- C_{ab} C_{cd} \Im\,t^c \Im\,t^d \right)\,.
\end{eqnarray}

\item For CY threefolds, the expressions for the K\"ahler potential and metric are given by:
\begin{eqnarray}
K &=& -\log \left(4 |X^0|^2 \left(\Im F_0 -\Im t^c \Re \,F_c\right)\right)\,,\\
G_{i\jb}&=&\frac{1}{4 \left(\Im F_0 -\Im t^c \Re \,F_c\right)} \frac{\partial t^a}{\partial z^i} \frac{\partial \ov{t}^b}{\partial \ov{z}^{\jb}} \left(2 \Im \tau_{ab}\, \left(\Im F_0 -\Im t^c \Re \,F_c\right)- \left( \Im F_a- \Im t^c \tau_{ac}\right)\left( \Im F_b- \Im t^d \ov{\tau}_{bd}\right) \right)\,.\nonumber\\
\end{eqnarray}

\end{itemize}

\end{thm}

This theorem is proved case by case in Sec.~\ref{metrics}.

\subsection{Structure of this work}
This work is structured as follows. In Sec.~\ref{background}, a review of the $\mathcal{N}=2$ superconformal algebra, its chiral ring, the $tt^*$ equations as well as their geometric realization in CY sigma models is provided. Sec.~\ref{flatconnections} proceeds with a discussion of the implications of the flatness of the $tt^*$ connection for the curvature of the K\"ahler metric of the moduli spaces of the CY d-folds, for d=1,2,3. In Sec.~\ref{metrics} the implications of the flatness of the $tt^*$ connection are exploited to obtain exact expressions for the K\"ahler metrics in terms of the special coordinates.  In Sec.~\ref{differentialrings}, the differential rings will be developed and in Sec.~\ref{liealgebra}, the generators of the differential rings will be used as parameters on the larger moduli space $\mathsf{T}$. On $\mathsf{T}$, vector fields along the differential ring generators will provide a basis for a Lie algebra which captures the full content of the $tt^*$  equations for these sigma models. In Sec.~\ref{examples}, examples will be given for the general structures developed in this paper.

\section{Moduli space of $\mathcal{N}=2$ superconformal algebra}\label{background}
In this section, some background material will be reviewed.  We will recall the $\mathcal{N}=2$ superconformal algebra and its chiral ring \cite{Lerche:1989uy} as well as the $tt^*$ equations \cite{Cecotti:1991me} and the geometric realization in terms of variation of Hodge structure. More details can be found in Refs.~\cite{Warner:1993zh,Cox:1999,Aspinwall:book}.

\subsection{$\mathcal{N}=2$ superconformal algebra and chiral ring}
\subsubsection{Superconformal algebra}
The $\mathcal{N}=2$ superconformal algebra is generated by the energy momentum tensor $T(z)$, two supercurrents $G^{\pm}(z)$, and a $U(1)$ current $J(z)$ of conformal weights $2,3/2,1$ respectively. $\pm$ denotes the $U(1)$ charge. The  boundary conditions for $G^{\pm}(z)$ are:
\begin{equation}
G^{\pm}(e^{2\pi i}z)=- e^{\mp 2\pi i a} G^{\pm}(z)  \, ,
\end{equation}
with a continuous real parameter $a$ which lies in the range $0 \le a < 1$. The superconformal algebras are isomorphic for different values of $a$. The map between these is called the spectral flow. For $a=0,1/2$ the representations are called Ramond (R) and Neveu-Schwarz (NS) sectors, respectively. The currents can be expanded in Fourier modes
\begin{equation}\label{modes}
T(z)= \sum_{n} \frac{L_n}{z^{n+2}}\, ,\quad G^{\pm}(z)= \sum_{n} \frac{G^{\pm}_{n \pm a}}{z^{n\pm a+\frac{3}{2}}}\, ,\quad J(z)=\sum_n \frac{J_n}{z^{n+1}} \, .
\end{equation}
The $\mathcal{N}=2$ superconformal algebra can be expressed in terms of the operator product expansion of the currents or by the commutation relations of their modes:
\begin{eqnarray}
\left[ L_m, L_n \right] &=& (m-n) L_{m+n} +\frac{c}{12} m (m^2-1) \delta_{m+n,0} \, , \nonumber \\
\left[  J_m,J_n\right] &=& \frac{c}{3}  m \delta_{m+n,0}\, ,\nonumber\\
\left[ L_n, J_m\right] &=& -m J_{m+n} \, ,\nonumber\\
\left[L_n,G_{m\pm a}^{\pm} \right] &=& \left(\frac{n}{2}-(m\pm a)\right) G^{\pm}_{m+n\pm a}\, , \nonumber\\
\left[ J_n,G_{m\pm a}^{\pm}\right] &=& \pm G_{n+m\pm a}^{\pm} \, ,\nonumber\\
\left\{  G_{n+a}^+,G^-_{m-a}\right\}&=& 2 L_{m+n} + (n-m+2 a)J_{n+m}+\frac{c}{3} \left((n+a)^2-\frac{1}{4}\right)\delta_{m+n,0}\, .
\end{eqnarray}

The following discussion will be in the NS sector ($a=1/2$). The highest weight states of the superconformal algebra, which are created by primary operators $\phi |0\rangle$ satisfy:
\begin{equation}
 L_n |\phi \rangle =0\, , \quad G^{\pm}_{r+1/2} |\phi\rangle=0\, ,\quad J_m|\phi\rangle =0 \,,\quad n,r,m>0\,.
\end{equation}
they are labeled by the eigenvalues of the zero index modes $L_0$ and $J_0$ 
\begin{equation}
 L_0|\phi\rangle =h_{\phi}|\phi \rangle \,, \quad J_0 |\phi \rangle = q_{\phi} |\phi\rangle\,.
\end{equation}

\subsubsection{The chiral ring}
Chiral states are a subset of the highest weight states, created by the chiral primaries which satisfy in addition:
\begin{equation}\label{chiral}
 G^{+}_{-1/2} |\phi\rangle =0\,.
\end{equation}
Anti-chiral states are annihilated by $G^-_{-1/2}$. 
Considering
\begin{equation}
 \langle \phi| \{G^-_{1/2},G^+_{-1/2}\} |\phi\rangle = || G^+_{-1/2} |\phi\rangle||^2= \langle \phi|2L_0-J_0|\phi\rangle\,\ge 0\,,
\end{equation}
implies 
\begin{equation}
h_{\phi}\ge \frac{q_{\phi}}{2} \, , 
\end{equation}
with equality holding for chiral states.  

Furthermore, in the operator product expansion of two chiral primary fields 
$\phi$ and $\chi$
\begin{equation}
 \phi(z) \chi(w) =\sum_i (z-w)^{h_{\psi_i} -h_{\phi}-h_{\chi}} \psi_i\, ,
\end{equation}
the $U(1)$ charges add $q_{\psi_i}=q_{\phi}+q_{\chi}$ and hence $h_{\psi_i}\geq h_{\phi}+h_{\chi}$. The operator product expansion has thus no singular terms and the only terms which survive in the expansion when $z\rightarrow w$ are the ones for which $\psi_i$ is itself chiral primary. It is thus shown that the chiral primary fields give a closed non-singular ring under operator product expansion, this is the chiral ring:
\begin{equation}\label{chiralring}
 \phi_i \phi_j = C_{ij}^k \phi_k \,.
\end{equation}

The charges of the ring are furthermore bounded, as can be seen from:
\begin{equation}
 \langle \phi| \{G^-_{3/2},G^+_{-3/2}\} |\phi\rangle = \langle \phi|2L_0-3 J_0+\frac{2}{3}c|\phi\rangle \geq 0\,,
\end{equation}
which implies $q\le \hat{c}$ for unitary theories, where $\hat{c}=c/3$

The $\mathcal{N}=(2,2)$ is then obtained by combining two copies of this algebra. This gives the $(c,c),(a,c),(a,a)$ and $(c,a)$ rings, where the latter two are charge conjugates of the first two.

\subsubsection{Ramond ground states}
If one considers the representation of the SCA with $a=0$ or the Ramond sector, the analog of the bound between conformal weight and charge becomes a bound on the conformal weight.
The analog of the chiral states become the subset of states satisfying:
\begin{equation}\label{ramond}
 G^{+}_{0} |\phi\rangle =0\,,
\end{equation}
these are the Ramond ground states. Considering
\begin{equation}
 \langle \phi| \{G^-_{0},G^+_{0}\} |\phi\rangle = || G^+_{0} |\phi\rangle||^2= \langle \phi|2L_0-\frac{c}{12} |\phi\rangle\,\ge 0\,,
\end{equation}
implies 
\begin{equation}
h_{\phi}\ge \frac{c}{24} \, , 
\end{equation}
with equality holding for the Ramond ground states.  A crucial fact is that the spectral flow isomorphism which can be implemented on the operators of the SCA leads to an identification of both the chiral ring as well as the anti-chiral ring with the finitely many Ramond ground states. This is a crucial insight in the development of the $tt^*$ equations.

\subsubsection{Twisting and $U(1)$ anomaly}

The Ramond ground states admit a translation into a topological field theory by considering the cohomology of $G_0^+$. Knowing the isomorphism between the chiral ring and the Ramond ground states one wants however also to take advantage of the ring structure. In order to do this one can consider $G_{-1/2}^+$ or $G_{-1/2}^-$ as operators and consider their cohomology which will give the chiral and anti-chiral rings respectively. In order to have a global operator on the $2d$ base space and in order to have a topological $T_{top}$ a further twist is needed. One possibility is to define:
\begin{equation}
T'=T_{top} = T +\frac{1}{2} \partial J\,.
\end{equation}
On the level of modes this becomes:
\begin{equation}
L'_m = L_m -\frac{1}{2}(m+1) J_m\,,
\end{equation}
 which has the effect of shifting $h'=h- \frac{q}{2}$. 
 Through this twist the operator $G_{-1/2}^+$ becomes a scalar and can be identified with:
 \begin{equation}
 \mathcal{Q} = G_{-1/2}^+\,,
 \end{equation}
where $\mathcal{Q}$ is an operator which squares to zero $\mathcal{Q}^2=0$ and which can be used to define a topological theory of cohomological type. Moreover $T_{top}$ becomes $\mathcal{Q}$-exact, as is required for the independence of the correlation functions on the insertion points. The twisting introduces furthermore an anomaly in the $U(1)$ current, which has the effect that correlation functions on the sphere are only non-vanishing when a total $U(1)$ charge $\hat{c}=d$ is inserted in the correlator. These correlators will play an important role, as they will be identified with the chiral ring structure constants. We will denote these by:
\begin{equation}
C_{i_1\dots i_d} = \langle \phi_{i_1}\dots \phi_{i_d}\rangle\,.
\end{equation}
If one wants to restrict to the anti-chiral ring on the other hand the twisting would be:
\begin{equation}
T'=T_{top} = T -\frac{1}{2} \partial J\,.
\end{equation}

\subsection{Moduli spaces and a flat connection}
\subsubsection{Deformations of SCFT}
We restrict the discussion to the $(c,c)$ and $(a,a)$ rings in the following. To deform a superconformal field theory one may add to the action operators:
\begin{equation}
\mathcal{S} \rightarrow \mathcal{S} + t^i \int d^2z\, \mathcal{O}_i + \bar{t}^{\ib} \int d^2 z\, \mathcal{O}_{\ib}\,.
\end{equation}
The operators $\mathcal{O}_i,\mathcal{O}_{\ib}$ should have $(q,\bar{q})=(0,0)$ and $(h,\bar{h})=(1,1)$, these can be constructed out of the chiral primary operators with $(|q|,|\bar{q}|)=(1,1)$. For example, starting with $\phi$  with $(q,\bar{q})=(1,1)$ from the $(c,c)$ ring one can construct:
\begin{equation}
 \phi^{(1)}(w,\overline{w})=\left[  G^-,\phi(w,\overline{w})\right]=\oint dz\, G^{-}(z)\phi(w,\overline{w})\, ,
\end{equation}
 which now has $h=1,q=0$. In the next step
\begin{equation}
\phi^{(2)}(w,\overline{w})=\left\{ \overline{G}^{-},\phi^{(1)}(w,\overline{w})\right\}=\oint d\overline{z}\, \overline{G}^{-}(\overline{z}) \phi^{(1)}(w,\overline{w})\, ,
\end{equation}
which has $h=\overline{h}=1$ and zero charge and is hence a truly marginal operator and can be used to perturb the action of the theory
\begin{equation}
\mathcal{O}_i= \phi_i^{(2)}\,, \quad \mathcal{O}_{\ib}=\phi_{\ib}^{(2)}\,, \quad i=1,\dots,n\,,
\end{equation}
where $n=\textrm{dim} \mathcal{H}^{(1,1)}$ denotes the dimension of the subspace of the Hilbert space of the theory containing the states which are created by the charge $(1,1)$ operators. A similar construction can be done for the $(a,c)$ chiral ring.  The deformations constructed in this way span a deformation space $\mathcal{M}$, the moduli space of the SCFT.


\subsubsection{$tt^*$ equations and deformation bundle}

In the following we will consider a subset of the states created by operators of the chiral ring and study how these vary as the parameters of the theory are changed. The subset in question will be the span of states obtained from the repeated action on the unique NS vacuum $|0\rangle$with the truly marginal operators. We adopt in the following the language of the $(c,c)$ ring, similar statements hold for the other rings as well. The truly marginal deformations in this case have charge $(q,\bar{q})=(1,1)$. We will think of these states as living in a bundle $\mathcal{H}\rightarrow \mathcal{M}$ over the moduli space $\mathcal{M}$ of the theory, we will call this bundle the deformation bundle. It should be noted that this will be a subbundle of the bundle of Ramond ground states which is usually considered in the $tt^*$ context. The bundle $\mathcal{H}$ can be locally decomposed into subbundles, using the charge grading, i.e.:
\begin{equation}
\mathcal{H}= \mathcal{H}^{0,0} \oplus \mathcal{H}^{1,1}\oplus \dots \oplus \mathcal{H}^{d,d}\,,
\end{equation}
where $\mathcal{H}^{p,q}$ denotes the subspace of states of charge $(p,q)$, it is now clear that it is a sub-bundle since it only carries states with $p=q$, since they are all created by the repeated action of operators of charge $(1,1)$ on the NS vacuum. Similar statements hold for the other rings as well.  We will label the states by $\alpha=1,\dots,N$, where $N$ is the dimension of the fiber of this bundle. These states form a representation of the chiral ring, so in particular the action of the truly marginal operators $\phi_i,i=1,\dots,n=\dim \mathcal{M}$ reads.
\begin{equation}
\phi_i |\mu \rangle= C_{i\mu}^{\nu} |\nu \rangle\,,
\end{equation}
the $C_{i\alpha}^{\beta}$ can be thought  of as the entries of an $N\times N$ matrix $\mathcal{C}_i$. Since the states in the deformation bundle also map to Ramond ground states by the spectral flow isomorphism, there is also a Berry connection, capturing the variation of the ground states as the parameters of the theory are varied and this defines a covariant derivative. Along the truly marginal directions acting on a state $| \mu \rangle$ this reads:
\begin{equation}
D_i |\mu \rangle := (\delta^{\mu}_{\nu}\partial_i - (A_i)^{\phantom{\nu}\mu}_{\nu}) |\mu\rangle\,.
\end{equation}

$\mathcal{H}$ has furthermore a Hermitian structure obtained by first considering the bundle of states by truly marginal operators of the $(a,a)$ ring. We will denote these by $|\bar{\mu}\rangle\,,\bar{\mu}=1,\dots,N$. Since also these states map to the Ramond ground states, the $tt^*$ geometry of Ref.~\cite{Cecotti:1991me} puts forward a hermitian metric on the moduli space, defined by:
\begin{equation}
g_{\mu\bar{\nu}}:= \langle \bar{\nu} | \mu \rangle\,.
\end{equation}
Similarly a Berry connection and the representation of the anti-chiral ring can be defined in this case. The $tt^*$ equations describe the combined variation with respect to both rings, exploiting the fact that both map to the Ramond ground states, the equations read \cite{Cecotti:1991me}:
\begin{eqnarray}\label{ttstar}
\left[ D_i,\mathcal{C}_{\ib}\right]&=& 0\,, \\
\left[ D_{\ib},\mathcal{C}_{i}\right]&=&0\,, \\
\left[ D_i,D_{\ib}\right] &=& - \left[ \mathcal{C}_{i},\mathcal{C}_{\ib}\right]\,, \\
\left[ D_i,\mathcal{C}_j\right] &=& \left[ D_j,\mathcal{C}_i\right]\,,\\
\left[ D_{\ib},\mathcal{C}_{\jb}\right] &=& \left[ D_{\jb},\mathcal{C}_{\ib}\right]\,,
\end{eqnarray}
this gives the connection:
\begin{equation}
\nabla_i= D_i + \frac{1}{\zeta} \, \mathcal{C}_i\,, \quad  \nabla_{\ib}= D_{\ib} + \zeta\, \mathcal{C}_{\ib}\,,
\end{equation}
which is flat:
\begin{equation}
\left[ \nabla_i,\nabla_{\ib}\right]=0\,,\quad \left[ \nabla_i,\nabla_j\right]=0\,,\quad \left[ \nabla_{\ib},\nabla_{\jb}\right]=0\,,
\end{equation}
where $\zeta$ is a spectral parameter which is important in the interpretation of the flat sections of the $tt^*$ connection. These correspond to D-branes in the target CY \cite{Hori:2000ck}. See also Refs.~\cite{Cecotti:2013mba,Vafa:2014lca} for more recent discussions of this parameter.

\subsection{Geometric realization as variation of Hodge structure}

We will consider the topological twist which leads to the B-model which has the $(c,c)$ ring as its states. When realized in terms of a NLSM these states translate into the cohomology of the target CY $d-$fold $\CY$. In particular the states of the deformation bundle get identified with:
\begin{equation}\label{chargedegree}
\mathcal{H}^{p,q} \equiv H^{0,q}(\CY,\wedge^p T\CY) \cong H^{d-p,q}(\CY)\,\,,
\end{equation}
where the last identification is made using the unique holomorphic $(d,0)$ form $\Omega$ of $\CY$.
The moduli space $\mathcal{M}$ on the B-side corresponds to the moduli space of complex structures of the target space $\CY$.  For a CY d-fold, the deformation bundle is the middle dimensional cohomology $H^{d}(\CY,\mathbbm{C})$. This space has a natural splitting once a given complex structure is chosen, i.e., at a specific point in the complex structure moduli space. The split is
\begin{equation} \label{Hodgesplit}
 H^d(\CY,\mathbbm{C}) \simeq \bigoplus_{p+q=d} H^{p,q}(\CY)\, .
\end{equation}
This split identifies the unique up to scale holomorphic $(d,0)$ form $\Omega$. It permits furthermore a natural notion of complex conjugation, namely $\overline{H^{p,q}(\CY)}=H^{q,p}(\CY)$. The term Hodge structure refers to $H^{d}(\CY,\mathbbm{C})$, together with the split (\ref{Hodgesplit}) and with a lattice given by $H^{d}(\CY,\mathbbm{Z})$ which generates $H^{d}(\CY,\mathbbm{C})$ upon tensoring with $\mathbbm{C}$; see Ref.~\cite{Cox:1999} for more details. The split (\ref{Hodgesplit}) does however not vary holomorphically when the complex structure moduli are varied. There is however a different split of the bundle which varies holomorphically over the moduli space of complex structures. This split is given by the Hodge filtration $F^{\bullet}(\CY)=\{ F^{p}(\CY)\}_{p=0}^d$, where the spaces in brackets are defined by
\begin{equation}\label{filtration}
H^d=F^0 \supset F^1\supset \dots  F^{d+1}=0\, , \quad F^{p}(\CY)=\bigoplus_{a\geq p} H^{a,d-a} (\CY)\, \subset H^d\, .
\end{equation}
To recover the splitting (\ref{Hodgesplit}) one can intersect with the anti-holomorphic filtration
\begin{equation}
 H^{p,q}(\CY)=F^{p}(\CY) \cap \overline{F^{q}(\CY)}\,. 
\end{equation}
Instead of a fixed target space, one can consider a variation family by varying the complex structure of $\CY$.
In this case the filtration is equipped with a flat connection $\nabla$ which is the Gauss-Manin connection with the property $\nabla F^p \subset F^{p-1}$, which is called Griffiths transversality. 

This property permits an identification of the derivatives of $\Omega(z) \in F^d$ with elements in the lower filtration spaces. The whole filtration can be spanned by taking multi-derivatives of the holomorphic $(d,0)$ form. $(d+1)$ order derivatives can then again be expressed by the elements of the basis, which is reflected by the fact that periods of $\Omega(z)$ are annihilated by a system of differential equations called the Picard-Fuchs (PF) equations. The PF equations capture the variation of Hodge structure which describes the geometric realization on the B-model side of the deformation of the $\mathcal{N}=(2,2)$ superconformal field theory and its chiral ring \cite{Lerche:1991wm}; see also Refs.~\cite{Ceresole:review,Alim:2012gq} for a review. Schematically the variation takes the form:
\begin{equation}
\begin{array}{ccccccc}
F^{d}&\xrightarrow{\nabla_B}&F^{d-1}&\xrightarrow{\nabla_B}&\dots&\xrightarrow{\nabla_B}&F^0\,.
\end{array}
\end{equation}

\section{Flat connections}\label{flatconnections}
In the following we will consider the $tt^*$ geometry of the realization of $\mathcal{N}=(2,2)$ as nonlinear sigma models into CY d-folds, this means $(\hat{c},\ov{\hat{c}})=(d,d)$. For CY threefolds, the $tt^*$ equations are equivalent to special geometry \cite{Strominger:1990pd}. While developing the structure we will adopt the language of the B-model, where the chiral ring translates to the variation of Hodge structure of the middle dimensional cohomology of the CY d-folds. We will bear in mind that analogous structures exist on the A-model side as well by mirror symmetry.

\subsection{Generalities}
The deformation bundle  is generated by acting on the unique NS ground state of the SCFT having $(q,\bar{q})=(0,0)$ with the marginal operators which have $(q,\bar{q})=(1,1)$. This translates geometrically via Eq.~(\ref{chargedegree}) to taking derivatives of the nowhere vanishing holomorphic $d-$form, which will denote by $\Omega$. This form is a section of the Hodge line bundle $\mathcal{L}\rightarrow \mathcal{M}$, where we denote by $\mathcal{M}$ the moduli space of complex structure of the CY d-fold which we denote by $\CY$. Here $\mathcal{L}$ is given by $F^d$ discussed earlier. This reflects the fact that in the CFT the ground state is unique up to scale. Furthermore let $\textrm{dim} \mathcal{M}=n$ and $z^i,\,i=1,\dots,n$ denote local coordinates in a given patch of $\mathcal{M}$. We will further denote by $\partial_i:=\frac{\partial}{\partial z^i}$ and $\partial_{\ib}:=\frac{\partial}{\partial \ov{z}^{\ib}}$.

We will introduce some structures which will be general for CY d-folds, of which we will only consider $d=1,2,3$. Similar structures for $d\ge 4$ were developed in Refs.~\cite{Greene:1993vm,Mayr:1996sh}.

We start by defining:
\begin{equation}
e^{-K}:= c_d \int_\CY \Omega \wedge \overline{\Omega}\, \in \Gamma(\mathcal{L} \otimes \ov{\mathcal{L}})\,,
\end{equation}
where $c_d$ is a normalization factor which depends on the conventions; we will give explicit values in the cases which we will consider. $K$ defines a K\"ahler potential and a K\"ahler metric can be given :
\begin{equation}
G_{i\ib} :=\partial_i \partial_{\ib} K\,,
\end{equation}
the Levi-Civita connection is then given by:
\begin{equation}\label{chris}
\Gamma_{ij}^k = G^{k\kb} \partial_i G_{j\kb}\,.
\end{equation}
To find an expression for the conection in $\mathcal{L}$ we require $D_i \Omega \in H^{(d-1,1)}(\CY)$, i.e. the connection part should cancel the $H^{d,0}$ parts in $D_i \Omega$, making the ansatz $D_i \Omega =\partial_i \Omega +A_i\Omega$. This means in particular that:
\begin{equation}
\int_{\CY} D_i \Omega \wedge \ov{\Omega} =0\,,
\end{equation}
by Hodge decomposition type considerations. This gives $A_i = \partial_i K =: K_i$. Similarly $D_{\ib} \ov{\Omega}=(\partial_{\ib}+ K_{\ib}) \ov{\Omega} \in H^{1,d-1}(\CY) $.

\begin{prop} We have the following
\begin{enumerate}
\item $ 
D_i e^{-K} = 0 =D_{\ib} e^{-K}\,.
$
\item 
$
G_{i\ib} = -c_d e^K \int_\CY D_i \Omega \wedge D_{\ib} \ov{\Omega}\,.
$
\item $D_i G_{j\jb}=0\,.$
\end{enumerate}
\end{prop}
\begin{proof}
\begin{enumerate}
\item This follows from the definitions. 
\item This follows from considering:
\begin{equation}
D_{\ib} D_i e^{-K}=0 = D_i D_{\ib}  \left(c_d \int_\CY \Omega \wedge \ov{\Omega} \right) = c_d \int_\CY D_i \Omega \wedge D_{\ib} \ov{\Omega} + G_{i\ib} e^{-K}\,.
\end{equation}
\item This follows from Eq.~\ref{chris}.
\end{enumerate}
\end{proof}

\subsection{Holomorphic limits}\label{holomorphiclimits}
In the following we will fix a notion of special coordinates and of holomorphic limits. Special coordinates (see Refs.~\cite{Strominger:1990pd,Ceresole:review} and references therein for CY threefolds) can be defined using the periods of the holomorphic $d-$form. In the following we will outline the similarities in the cases $d=1,2,3$ which we consider, a case by case discussion follows in subsequent sections. We fix a basis $\gamma^{\mu} \in H_d(\CY,\mathbbm{Z})\,, \mu=0,\dots, N-1$ and define the periods as the integrals:
\begin{equation}
\pi^{\mu}= \int_{\gamma^{\mu}} \Omega\,.
\end{equation}
Choosing $\gamma^{\mu}$ for $\mu=0,\dots,n$ non-intersecting cycles we obtain $X^I=\pi^I, I=0,1,\dots,n$ projective coordinates for $\mathcal{M}$, the moduli space of complex structures of $\CY$. We define special coordinates (in a patch where $X^0\ne0$) by:
\begin{equation}
t^a= \frac{X^a}{X^0}\,,\quad a=1,\dots,n\,.
\end{equation}

We furthermore adopt the discussion of holomorphic limits of the ingredients of the special K\"ahler geometry of Ref.~\cite{Bershadsky:1993cx}, by defining the holomorphic limit of $e^{-K}$ by:
\begin{equation}
e^{-K}|_{\textrm{hol}} := \mathsf{h}_0\, X^0\,,
\end{equation}
where $\mathsf{h}_0$ is a constant. We furthermore introduce the non-holomorphic coordinates:
\begin{equation}
t^a_{nh} := \mathsf{h}^{a\jb} K_{\jb}\,,
\end{equation}
where $\mathsf{h}_{a\jb}$ denote the entries of a constant matrix which is chosen such that the holomorphic limit of the latter is given by:
\begin{equation}
t^{a}_{nh} |_{hol} = t^a\,,
\end{equation}
it follows that
\begin{equation}
G_{i\jb}|_{hol} := \mathsf{h}_{a\jb}\frac{\partial t^a}{\partial z^i}\,.
\end{equation}
We will furthermore define:
\begin{equation}
\partial_a := \frac{\partial}{\partial t^a_{nh}}\,
\end{equation}
and, depending on the context, we may also think of $\partial_a$ as the derivative with respect to the special coordinate. We will in general also drop the subscript in $t^a_{nh}$. The coordinates $t^a_{nh}$ appear in the discussion of the canonical coordinates of Ref.~\cite{Bershadsky:1993cx}. It should be stressed that the discussion of the holomorphic limit is mathematically far from obvious, see Refs.~\cite{Bershadsky:1993cx,Zhou:2013hpa} for the discussion in the case of CY threefolds.

\subsection{Calabi-Yau onefolds}
We will start with the geometric realization of the SCA with $\hat{c}=1$ as a NLSM into a CY onefold. We will have the elliptic curve as a non-trivial compact example in mind. The moduli space of complex structures $\mathcal{M}$ is one dimensional and we will use $z$ as an algebraic local coordinate. Using the holomorphic $(1,0)$ form $\Omega$ which is a section of $\mathcal{L}\rightarrow \mathcal{M}$, we define the K\"ahler potential
\begin{equation}
e^{-K}= i  \int \Omega \wedge \overline{\Omega}\,, \in \Gamma(\mathcal{L} \otimes \ov{\mathcal{L}}). 
\end{equation}
The non-vanishing correlation function on the sphere has to have $\hat{c}=1$, this becomes a one point function of the marginal operator, we define:
\begin{equation}
C_z =\langle \phi_z \rangle =  -\int_\CY \Omega \wedge D_z \Omega\,,
\end{equation}
where $D_z \Omega = \partial_z \Omega +K_z \Omega $. 

\begin{prop}\label{propell} We have the following equations:
\begin{enumerate}
\item $D_z \Omega = -i e^K C_z\,\ov{\Omega}\,,\quad D_{\bar{z}} \ov{\Omega}= i e^K \ov{C}_{\bar{z}} \,\Omega$\,,
\item $
G_{z\bar{z}}= e^K C_z \ov{C}_{\bar{z}}\,,
$
\item 
$
D_z C_z=0\,,\quad D_{\bar{z}} \ov{C}_{\bar{z}}=0\,,
$
\end{enumerate}
\end{prop}

\begin{proof}
\begin{enumerate}
\item $D_z \Omega$ defines $(0,1)$ form, one can therefore make the Ansatz
$D_z \Omega = A_z\, \ov{\Omega}.$
To compute the coefficient $A_z$ we use Griffiths transversality and obtain
\begin{equation}
C_z =- \int_\CY \Omega \wedge D_z \Omega = -A_z\, \int_\CY \Omega \wedge \ov{\Omega}=  i A_z e^{-K}\,.
\end{equation} 
$D_{\bar{z}} \ov{\Omega}$ follows similarly.
\item This follows from 
$G_{z\bar{z}}= -i e^{K} \int_\CY D_z \Omega \wedge D_{\bar{z}} \ov{\Omega}\,,$ and $1$.
\item This follows from:
$
0=D_z G_{z\bar{z}}= D_z e^{2K} C_z \ov{C}_{\bar{z}}= e^{2k} \ov{C}_{\bar{z}} D_z C_z\,.
$
\end{enumerate}
\end{proof}

We have the following 
\begin{equation}
D_z \left(\begin{array}{c} \Omega \\ \ov{\Omega}\end{array}\right) = \underbrace{\left( \begin{array}{cc} 0 & -ie^K C_z\\ 0 &0 \end{array}\right)}_{:=\mathcal{C}_z} \left( \begin{array}{c} \Omega \\ \ov{\Omega} \end{array}\right)\,,
\end{equation}
and
\begin{equation}
D_{\bar{z}} \left(\begin{array}{c} \Omega \\ \ov{\Omega}\end{array}\right) = \underbrace{\left( \begin{array}{cc} 0 & 0\\ ie^K C_{\bar{z}} &0 \end{array}\right)}_{:=\mathcal{C}_{\bar{z}}} \left( \begin{array}{c} \Omega \\ \ov{\Omega} \end{array}\right)\,,
\end{equation}
in total these equations, together with $D_z \ov{C}_{\bar{z}}=0=D_{\bar{z}}C_z$, are manifestly equivalent to the $tt^*$ equations \ref{ttstar}. We can furthermore spell out the constraint on the curvature of the K\"ahler metric given by the flatness of the $tt^*$ connection. This is the analog of the well known constraints for threefolds, see Refs.~\cite{Strominger:1990pd,Bershadsky:1993cx}. This leads to the first part of Theorem~\ref{curvaturethm}:
\begin{thm*}
The curvature of the K\"ahler metric on $\mathcal{M}$ can be expressed as:
\begin{equation}\label{ellcurvature}
\partial_{\bar{z}} \Gamma_{zz}^z= G_{z\bar{z}}+ e^{2K} C_{z} \ov{C}_{\bar{z}}\,.
\end{equation}
\end{thm*}
\begin{proof}
We first compute:
\begin{equation}
[D_z,D_{\bar{z}}] D_z \Omega = D_z (G_{z\bar{z}} \Omega) - \partial_{\bar{z}} (\partial_z-\Gamma_{zz}^z+K_z) D_z \Omega=( \partial_{\bar{z}} \Gamma_{zz}^z - G_{z\bar{z}})\,,
\end{equation}
and now compute it again using $D_z\Omega = -ie^K C_z \ov{\Omega}$
\begin{equation}
[D_z,D_{\bar{z}}] D_z \Omega= [D_z,D_{\bar{z}}] (-ie^K C_z \ov{\Omega})= e^{2K} C_z \ov{C}_{\bar{z}} D_z \Omega\,.
\end{equation}
\end{proof}


\subsection{Calabi-Yau twofolds}\label{twofolds}
We proceed with the discussion of NSLM with $\hat{c}=2$, which correspond to target spaces which are CY twofolds. K3 surfaces are a non-trivial compact example. In these cases, an enhanced $\mathcal{N}=(4,4)$ SCA can be considered \cite{Eguchi:1988vra}, which is obtained from the $\mathcal{N}=(2,2)$ algebra by including the generators of additional symmetries.  The discussion in the following will however follow the $\mathcal{N}=(2,2)$ setup, which is also the natural arena of mirror symmetry for lattice polarized $K3$ manifolds \cite{Dolgachev:1996}, see also Refs.~\cite{Aspinwall:1994rg,Aspinwall:1996mn,Hosono:2000eb}. 

Lattice polarized K3 manifolds were defined in Ref.~\cite{Dolgachev:1996}, we give a brief outline here following Ref.~\cite{Hosono:2000eb}. The K3 lattice $\Lambda_{K3}$ is given by:
\begin{equation}
\Lambda_{K3}= E_8(-1) \oplus E_8(-1) \oplus H^{\oplus3}\,,
\end{equation}
where $H$ is the rank 2 hyperbolic lattice. A lattice polarized K3 is defined by a lattice $M$ of rank $r$ of signature $(1,r-1)$ which admits a primitive embedding into $\Lambda_{K3}$. Given a lattice $M$, the $M$-polarized K3 surface is defined to be a K3 surface whose Picard lattice is $M$. The orthogonal complement of the Picard lattice gives the transcendental lattice of the K3 surface and gives up to a factor $H$ the Picard lattice of the mirror K3. The discussion of the variation of Hodge structure in this work will be concerned with the part of $H^{2}(\CY,\mathbbm{C})$ coming from the transcendental lattice $T(\CY)$ of the mirror K3 surface $\CY$, this lattice has rank $2+n$ when $n$ is the rank of the Picard lattice of $\check{\CY}$ and $n=\dim \mathcal{M}$, the dimension of the moduli space of complex structures of $\CY$.

Using the holomorphic $(2,0)$ form $\Omega \in \Gamma(\mathcal{L})$ , we define the K\"ahler potential
\begin{equation}
e^{-K}=  \int \Omega \wedge \overline{\Omega}\,, \in \Gamma(\mathcal{L} \otimes \ov{\mathcal{L}})
\end{equation}
We further define a holomorphic section of $\mathcal{L}^2\otimes \textrm{Sym}^2 T^*\mathcal{M}$, which corresponds to the non-vanishing two point function on the sphere, by
\begin{equation}
C_{ij} =\langle \phi_i \phi_j\rangle= -\int_\CY \Omega \wedge D_i D_j \Omega\,,
\end{equation}

\begin{prop}The following equations hold:
\begin{enumerate}
\item $\int_\CY D_i \Omega \wedge D_j \Omega= C_{ij}\,.$
\item $D_i \Omega = -e^K G^{j\ib} C_{ij} D_{\ib}\ov{\Omega}\,,\quad D_{\ib} \ov{\Omega}= -e^{K} G^{i\jb} \ov{C}_{\ib\jb} D_i \Omega\,.$
\item $D_i D_j \Omega= e^K C_{ij} \ov{\Omega}$.
\item $D_i C_{jk}=0\,,\quad D_{\ib} \ov{C}_{\jb\bar{k}}=0\,.$
\end{enumerate}
\end{prop}

\begin{proof}
\begin{enumerate}
\item Since $\int_\CY \Omega \wedge D_i \Omega =0 $ by Hodge decomposition type, we have
\begin{equation}
D_{j}  \int_\CY \Omega \wedge D_i \Omega = \int_\CY D_j\Omega \wedge D_i \Omega + \int_\CY \Omega \wedge D_j D_i \Omega =\int_\CY D_i \Omega \wedge D_j \Omega - C_{ij}=0\,.
\end{equation}
\item  $D_i\Omega$ should define a $(1,1)$ form by Griffiths transversality, we therefore make the Ansatz
$D_i \Omega = A_i^{\bar{\imath}}\,D_{\bar{\imath}} \ov{\Omega}$ and use the following:
\begin{align}
\int_{\CY} D_{i}\Omega \wedge D_{j} \Omega=A_{i}^{\ib} \int_{\CY} D_{\ib} \ov{\Omega} \wedge D_{j} \Omega &= C_{ij} \,,
\end{align}
And hence
\begin{equation}
A_i^{\ib} = -e^K G^{j\ib} C_{ij}\,,\quad \textrm{likewise} \quad D_{\ib} \ov{\Omega} = A_{\ib}^i\,D_i \Omega\,,\quad A^i_{\ib} = -e^K G^{i\jb} \ov{C}_{\ib\jb}\,.
\end{equation}
\item By Griffiths' transversality we have: $D_i D_j \Omega= A_{ij} \ov{\Omega} \in H^{0,2}$, we compute:
\begin{equation}
-C_{ij}=\int_\CY \Omega \wedge D_i D_j \Omega =  A_{ij} \int_\CY \Omega \wedge \ov{\Omega} = - A_{ij} e^{-K}
\end{equation}
 and hence:
 \begin{equation}
 A_{ij}= e^K \,C_{ij} 
 \end{equation}
\item  We can show:
 \begin{eqnarray}
 D_k \int_\CY D_i \Omega \wedge D_j \Omega = \int_\CY D_k D_i \Omega \wedge D_j \Omega + \int_\CY D_i \Omega \wedge D_k D_j \Omega =\nonumber\\
 e^K C_{ki} \int_\CY \ov{\Omega} \wedge D_j \Omega + e^K C_{kj} \int_\CY \ov{\Omega} \wedge D_i \Omega =0
 \end{eqnarray}
 and hence:
 \begin{equation}\label{YukflatK3}
 D_i C_{jk}=0\,.
 \end{equation}
\end{enumerate}
\end{proof}
These equations can be written in the following form:
\begin{eqnarray}
D_m \left(\begin{array}{c} \Omega \\ D_i \Omega\\ \ov{\Omega}\end{array}\right) &=& \underbrace{\left( \begin{array}{ccc} 0 & \delta^j_m &0\\ 0 &0&e^K C_{mi} \\ 0&0&0 \end{array}\right)}_{:=\mathcal{C}_m}  \left(\begin{array}{c} \Omega \\ D_j \Omega\\ \ov{\Omega}\end{array}\right) \,,
\\
D_{\bar{m}} \left(\begin{array}{c} \Omega \\ D_i \Omega\\ \ov{\Omega}\end{array}\right) &=& \underbrace{\left( \begin{array}{ccc} 0 & 0 &0\\ G_{\bar{m}i}&0 &0\\ 0&-e^K G^{j\jb} \ov{C}_{\bar{m\jb}}&0 \end{array}\right)}_{:=\mathcal{C}_{\bar{m}}}  \left(\begin{array}{c} \Omega \\ D_j \Omega\\ \ov{\Omega}\end{array}\right) \,,
\end{eqnarray}
which are again an explicit realization of the $tt^*$ equations \ref{ttstar}. This leads to the second part of Theorem~\ref{curvaturethm}:

\begin{thm*}The constraint on the curvature of the K\"ahler metric is given by:
\begin{equation}\label{curvatureK3}
\partial_{\bar{m}} \Gamma_{im}^k=\delta^{k}_i G_{m\bar{m}} +\delta_m^k G_{i\bar{m}} + e^{2K} C_{im} \, \ov{C}_{\bar{m}\bar{k}} G^{k\kb}\,.
\end{equation}
Furthermore, the metric satisfies:
\begin{equation}\label{K3algrel1}
 G_{m\bar{m}} = - C_{mj} \,\ov{C}_{\bar{m}\jb}\, G^{j\jb}\, e^{2K}\,, 
 \end{equation}
\end{thm*}
\begin{proof}
We compute the L.H.S and R.H.S of the commutator acting on $(\Omega \,\,D_i\Omega\,\, \ov{\Omega})^{Tr}$
\begin{equation}
\left[ D_m,D_{\bar{m}} \right] = - \left[ \mathcal{C}_m,\mathcal{C}_{\bar{m}} \right]
\end{equation}

In the first line we verify:
\begin{equation}
(D_m D_{\bar{m}}- D_{\bar{m}} D_m ) \Omega = -G_{m\bar{m}} \Omega\,,
\end{equation}
the second line gives:
\begin{equation}
(D_m D_{\bar{m}}- D_{\bar{m}} D_m )D_i \Omega = (e^{2K} C_{mi} \ov{C}_{\bar{m}\kb}G^{k\kb} +G_{\bar{m}i} \delta^k_m) D_k \Omega\,,
\end{equation}
writing out the L.H.S and rearranging gives the expression for the curvature.
The third line gives:
\begin{equation}
(D_m D_{\bar{m}}- D_{\bar{m}} D_m ) \ov{\Omega}= G_{m\bar{m}} \ov{\Omega}= - C_{mj} \,\ov{C}_{\bar{m}\jb}\, G^{j\jb}\, e^{2K} \, \ov{\Omega}, 
\end{equation}
which gives the expression for the metric.
\end{proof}
\subsection{Calabi-Yau threefolds}
The case of $\hat{c}=3$ corresponds to sigma models into CY threefolds. The flatness of the $tt^*$ connection is equivalent to special geometry \cite{Strominger:1990pd}. The manipulations given in the following can be found in many places; see e.~g.~Refs.~\cite{Candelas:1990pi,Bershadsky:1993cx}.
Using $\Omega \in H^{3,0}(\CY)$, we define:
\begin{equation}
e^{-K}:= i\int \Omega \wedge \ov{\Omega}\,,
\end{equation}
which gives the K\"ahler metric:
\begin{equation}
G_{i\ib}= \partial_i \partial_{\ib}= -ie^K \int_{\CY} D_i\Omega \wedge D_{\ib} \ov{\Omega}\,.
\end{equation}
We further define a holomorphic section of $\mathcal{L}^2\otimes \textrm{Sym}^3 T^*\mathcal{M}$, which corresponds to the non-vanishing three-point function on the sphere, by
\begin{equation}\label{threefoldYukawa}
C_{ijk} =\langle \phi_i \phi_j \phi_k\rangle= -\int_\CY \Omega \wedge D_i D_j D_k\Omega\,,
\end{equation}
We have the following
\begin{prop} The following equations hold:
\begin{enumerate}
\item $\int_{\CY} D_i \Omega \wedge D_j D_k \Omega= C_{ijk}\,.$
\item $D_i D_j\Omega = -i e^K G^{k\kb} C_{ijk} D_{\kb} \ov{\Omega}\,,\quad D_{\ib} D_{\jb} \ov{\Omega} = i e^K G^{k\kb} \ov{C}_{\ib\jb\kb} D_{k} \Omega\,.$
\item $D_i D_{\jb} \ov{\Omega}= G_{i\jb} \Omega\,.$
\end{enumerate}
\end{prop}
\begin{proof}
\begin{enumerate}
\item $\int_{\CY} \Omega \wedge D_{i}D_j \Omega=0$ by type considerations, this follows from $D_k \int_{\CY} \Omega \wedge D_{i}D_k \Omega =0$ and Eq.~\ref{threefoldYukawa}.
\item $D_{i}D_j \Omega = A_{ij}^{\kb} D_{\kb} \ov{\Omega} \in H^{1,2}(\CY)$ by Hodge type considerations, the coefficients are found using 1.
\item This is a straightforward computation.
\end{enumerate}
\end{proof}

These equations can be written in the following form:
\begin{equation}
D_m \left(\begin{array}{c} \Omega \\ D_i \Omega\\ D_{\ib} \ov{\Omega} \\\ov{\Omega}\end{array}\right) = \underbrace{\left( \begin{array}{cccc} 0 & \delta^j_m &0&0\\ 0 &0&-ie^K C_{mij} G^{j\jb}&0 \\ 0&0&0&G_{m\ib}\\0&0&0&0 \end{array}\right)}_{:=\mathcal{C}_m}  \left(\begin{array}{c} \Omega \\ D_j \Omega\\ D_{\jb}\ov{\Omega}\\\ov{\Omega}\end{array}\right) \,,
\end{equation}
and
\begin{equation}
D_{\bar{m}} \left(\begin{array}{c} \Omega \\ D_i \Omega\\ D_{\ib} \ov{\Omega} \\\ov{\Omega}\end{array}\right) = \underbrace{\left( \begin{array}{cccc} 0 & 0 &0&0\\ G_{\bar{m}i}&0 &0&0\\ 0&ie^K G^{j\jb} \ov{C}_{\bar{m}\ib\jb}&0&0\\0&0&\delta^{\jb}_{\bar{m}}&0 \end{array}\right)}_{:=\mathcal{C}_{\bar{m}}}  \left(\begin{array}{c} \Omega \\ D_j \Omega\\ D_{\jb} \ov{\Omega} \\ \ov{\Omega}\end{array}\right) \,,
\end{equation}
which are again an explicit realization of the $tt^*$ equations \ref{ttstar}.
This leads to the well known constraint of special geometry of CY threefolds \cite{Strominger:1990pd,Bershadsky:1993cx}:
\begin{thm*}The constraint on the curvature of the K\"ahler metric is given by:
\begin{equation}\label{curvaturethreefold}
\partial_{\bar{m}} \Gamma_{im}^k=\delta^{k}_i G_{m\bar{m}} +\delta_m^k G_{i\bar{m}} - e^{2K} C_{ijm} \, \ov{C}_{\bar{m}\bar{k}\jb} G^{j\jb}G^{k\kb}\,.
\end{equation}
\end{thm*}
\begin{proof}
This follows by similar considerations as in the one and two-fold cases.
\end{proof}
 

\section{Special coordinates and $tt^*$ metrics}\label{metrics}

In the following, special coordinates on the moduli spaces $\mathcal{M}$ of the target spaces $\CY$ will be introduced. The constraints arising from the flatness of the $tt^*$ connection developed in the previous section will be enough to recover a solution of the $tt^*$ geometry for the elliptic curve, given in Ref.~\cite{Cecotti:1990wz} and present an analogous equation for metric of the ground states for the NLSM with CY two- and threefolds as target manifolds.\footnote{It should be noted that all the ingredients for CY threefolds are known, whereas to the best of the author's knowledge the twofold solution is new.}

\subsection{Elliptic curve} \label{ellflat}
We choose a symplectic basis of $\{\alpha,\beta\} \in H^1(\CY,\mathbbm{Z})$ such that:
\begin{equation}
\int_\CY \alpha \wedge \beta =1\,,\quad \int_\CY \alpha \wedge \alpha =0\,,\quad \int_\CY \beta \wedge \beta=0\,.
\end{equation}
let furthermore $\{A,B\}$ denote a dual basis of cycles in $H_1(\CY,\mathbbm{Z})$ such that:
\begin{equation}
\int_\CY \alpha \wedge \beta= \int_A \alpha = \int_B \beta=1\,.
\end{equation} 
We will denote by $\{\pi^0,\pi^1\}$ the periods of $\Omega$ over the cycles $\{A,B\}$:
\begin{equation}
\int_A \Omega= \pi^0\,,\quad \int_B \Omega= \pi^1\,,
\end{equation}
we can now write:
\begin{equation} \label{omegaell}
\Omega = \pi^0 \alpha + \pi^1 \beta\,.
\end{equation}
We introduce the special coordinate $t= \pi^1/\pi^0$ and a normalized $\tilde{\Omega}:= (\pi^0)^{-1} \Omega$, s.t.:
\begin{equation}
\tilde{\Omega}= \alpha + t \beta\,,
\end{equation}

Using these we get the following:
\begin{thm*}
The K\"ahler potential $K$ and metric (solution of $tt^*$ equations) takes the form:
\begin{equation}
K =- \log\left(2 |\pi^0|^2 \, \Im \t\,\right), \quad G_{z\bar{z}}= \frac{1}{4|\Im\, t|^2} \left| \frac{\partial t}{\partial z}\right|^2\,,
\end{equation}
where $\Im t= \frac{1}{2i} (t-\bar{t}).$
\end{thm*}
\begin{proof}
We first compute $e^{-K}$ using Eq.~(\ref{omegaell}), the metric then follows immediately. 
\end{proof}
This agrees with the solution for the metric of the $tt^*$ geometry given in Refs.~\cite{Cecotti:1990wz,Cecotti:1991me}.

\subsection{Calabi-Yau twofolds}
The following discussion is for lattice polarized K3 manifolds \cite{Dolgachev:1996}. See in particular Ref.~\cite{Hosono:2000eb} and references therein for a discussion of the periods in this case. We chose a basis $\{ \alpha_0,\alpha_1,\dots,\alpha_{n+1}\}$ of the transcendental lattice $T(\CY) \subset H^2(\CY,\mathbbm{Z})$ where $n=\textrm{dim} \mathcal{M}$ such that:
\begin{equation}
\int_\CY \alpha_0 \wedge \alpha_{n+1} =-1\,\quad \int_\CY \alpha_{a} \wedge \alpha_b = C_{ab}\,,\quad  a,b=1,\dots,n\,,
\end{equation}
we denote by $\pi^I,I=0,\dots,n+1$ the periods of $\Omega$ over a dual basis of cycles in $H_2(\CY,\mathbbm{Z})$ such that:
\begin{equation} \label{k3periods}
\Omega= \sum_{I=0}^{n+1} \pi^I \alpha_{I}\,,
\end{equation}
we further introduce $\tilde{\Omega}= (\pi^0)^{-1} \Omega$ and the special coordinates:
\begin{equation}
t^a= \frac{\pi^a}{\pi^0}\,,\quad a=1,\dots, n\,,
\end{equation}
we obtain from
\begin{equation}
\int_\CY \Omega \wedge \Omega =0\,, 
\end{equation}
an expression for 
\begin{equation}
\omega^{n+1}= \frac{1}{2} C_{ab}t^a t^b\, \pi^0\,.
\end{equation}

Putting these equations together we obtain the second part of Theorem \ref{metricthm}:
\begin{thm*}
The K\"ahler potential $K$ and K\"ahler metric $G_{i\jb}$ on $\mathcal{M}$ (solution to the $tt^*$ equations) are given by:
\begin{eqnarray}
K &=& -\log \left( 2 |\pi^0|^2 C_{ab} \Im t^a \Im t^b \right)\,,\\
G_{i\jb} &=&\frac{1}{2  (C_{cd} \Im\,t^c \Im\,t^d)^2} \frac{\partial t^a}{\partial z^i} \frac{\partial \ov{t}^{b}}{\partial z^{\jb}} \left( 2 C_{ad} C_{bc} \Im\,t^c \Im\,t^d- C_{ab} C_{cd} \Im\,t^c \Im\,t^d \right)\,.
\end{eqnarray}
\end{thm*}
\begin{proof}
$e^{-K}$ is computed using Eq.~\ref{k3periods}, the metric follows immediately.
\end{proof}


\subsection{Calabi-Yau threefolds}
We pick a symplectic basis $\{\alpha_I,\beta^{J}\}, I,J=0,\dots,h^{2,1}$ of $H^{3}(\CY,\mathbbm{Z})$, satisfying
\begin{equation}
\int_{\CY }\alpha_I \wedge \beta^J = \delta^J_I\,, 
\end{equation}
let furthermore $\{A^I,B_J \}$  denote a dual basis of $H_{3}(\CY,\mathbbm{Z})$.  The holomorphic $(3,0)$ form can be expanded in this basis using the periods:
\begin{equation}
X^I :=\int_ {A^I}\Omega,\ \  \mathcal{F}_J:=\int_{B_J}\Omega\, ,\quad\textrm{and hence}\quad \Omega = X^I \alpha_I + \mathcal{F}_J \beta^J\,.
\end{equation}

The periods $X^I$ are identified with projective coordinates on $\mathcal{M}$ and 
$\mathcal{F}_J$ with derivatives of a homogeneous function $\mathcal{F}(X^I)$ of weight 
2 such that $\mathcal{F}_J=\frac{\partial \mathcal{F}(X^I)}{\partial X^J}$. 
The special coordinates are given by
\begin{equation} 
t^a=\frac{X^a}{X^0}\, ,\quad a=1,\dots ,n.
\end{equation}
The normalized holomorphic $(3,0)$ form $ \tilde{\Omega}_t :=(X^0)^{-1} \Omega_z$ has the periods:
\begin{equation}
\label{cmsa2014}
\int_{ A^I, B_J}\tilde \Omega_t= \left (1,t^a, F_b(t), 2F_0(t)-t^c F_c(t) \right) \,,
\end{equation}
where $$F_0(t)= (X^0)^{-2} \mathcal{F} \quad \textrm{and} \quad F_a(t):=\partial_a F_0(t)=\frac{\partial F_0(t)}{\partial` t^a}.$$
$F_0(t)$ is the called the prepotential and 
\begin{equation}
\tau_{ab}:=\partial_a \partial_b F_0(t)\,,\quad C_{abc}=\partial_a \partial_b \partial_c F_0(t)\,.
\end{equation} 
See Ref.~\cite{Ceresole:review,Alim:2012gq} for more details.
We obtain the following:
\begin{prop} The expressions for the K\"ahler potential and metric are given by:
\begin{eqnarray}
e^{-K} &=& 4 |X^0|^2 \left(\Im F_0 -\Im t^c \Re \,F_c\right)\,,\\
G_{i\jb}&=&\frac{1}{4 \left(\Im F_0 -\Im t^c \Re \,F_c\right)} \frac{\partial t^a}{\partial z^i} \frac{\partial \ov{t}^b}{\partial \ov{z}^{\jb}} \left(2 \Im \tau_{ab}\, \left(\Im F_0 -\Im t^c \Re \,F_c\right)- \left( \Im F_a- \Im t^c \tau_{ac}\right)\left( \Im F_b- \Im t^d \ov{\tau}_{bd}\right) \right)\,.\nonumber\\
\end{eqnarray}
\end{prop}
\begin{proof}
This is a straightforward computation.
\end{proof}

\section{Differential rings} \label{differentialrings}
In this section, a special set of functions will be put forward, which form a differential ring under taking holomorphic derivatives. These were developed for the mirror quintic in Ref.~\cite{Yamaguchi:2004bt} and generalized to arbitrary CY threefolds in Ref.~\cite{Alim:2007qj} and further studied in Ref.~\cite{Hosono:2008np}, where the analogous structure of elliptic curves was discussed. The matching with the functions of Ref.~\cite{Mov11}, provided for CY threefolds in Ref.~\cite{Alim:2014dea}, was obtained by enhancing the differential ring to the special differential ring, given for one dimensional moduli spaces in Ref.~\cite{Alim:2013eja} and generalized in Ref.~\cite{Alim:2014dea}. The special differential ring has compared to the differential ring additional generators with values in $\mathbbm{C}^*$, which parameterize the choice of trivialization of the Hodge line bundle $\mathcal{L}\rightarrow \mathcal{M}$ as well as the choice of special coordinates. We will start with the elliptic curve and then put forward the analogous structure for lattice polarized $K3$ manifolds and then include the differential rings for threefolds, following Refs.~\cite{Alim:2007qj,Alim:2014dea}. 

\subsection{Elliptic curve}\label{diffringell}
We introduce $S \in \Gamma_{C^{\infty}}(\mathcal{L}^{-2})$ which satisfies:
\begin{equation}
\partial_{\bar{z}}  S = e^{2K} \ov{C}_{\bar{z}},
\end{equation}
We now adopt techniques of Refs.~\cite{Alim:2007qj,Hosono:2008np} to compute:
\begin{equation}
\partial_{\bar{z}} D_z S= \partial_{\bar{z}}\partial_{z} S -2\partial_{\bar{z}}(K_{z} S)=-2G_{z\bar{z}} S= -2 C_z \partial_{\bar{z}} S \,S\,,
\end{equation}
which can be integrated to give \cite{Hosono:2008np}:
\begin{equation}
D_z S = -C_z S^2+ h_z\,,
\end{equation}
where $h_z$  is a holomorphic section of $\mathcal{L}^{-2} \otimes T^*\mathcal{M}$ depending on the choice of $S$.
As a defining equation of $S$ we can integrate Eq.~(\ref{ellcurvature}):
\begin{equation}
\partial_{\bar{z}} \Gamma_{zz}^z= G_{z\bar{z}}+ e^{2K} C_{z} \ov{C}_{\bar{z}}= 2 G_{z\bar{z}}\,,
\end{equation}
to 
\begin{equation}
\Gamma_{zz}^z= 2 L_z + s_z\,,
\end{equation}
where $s_z$ is a holomorphic function parameterizing the kernel of $\partial_{\bar{z}}$. And $L_z+l_z=K_z$, where $l_z$ can be chosen to be a rational function in $z$. We furthermore have an algebraic relation:
\begin{equation}
G_{z\bar{z}}= C_z \partial_{\bar{z}} S\,,
\end{equation}
which can be integrated to:
\begin{equation}\label{ltchoice}
L_z = C_z S+ \mathcal{H}_z \,.
\end{equation}
we will make the choice  $\mathcal{H}_z=0$.
Similarly the following can be computed:
\begin{equation}\label{Leq}
D_z L_z =-L_z^2 + k_{zz}\,,
\end{equation}
where $k_{zz}$ is a holomorphic function which depends on the choice of $L_z$ and which is computed from this equation. In order to make contact with the Lie algebra discussion of Ref.~\cite{Mov12a} we introduce further generators:
\begin{dfn}
The generators of the special differential ring are given by:
\begin{eqnarray}
g_0 &:=& \mathsf{h}_0^{-1}\,e^{-K}\,,\\
g_z^t &:=& \mathsf{h}^{t\bar{z}} G_{z\bar{z}} e^{-K}\,, \\
L_t &:=& g_0 (g^{-1})_t^z L_z\,,\\
T &:=&  g_0^2 S\,.
\end{eqnarray}
\end{dfn}
These satisfy the following:
\begin{prop}{Differential ring relations}
\begin{eqnarray}
\partial_t g_0 &=& -(L_t+l_t) g_0\,,\\
\partial_t g_z^t &=& g_z^t (L_t-l_t+s_t) \,,  \,\\
\partial_t L_t &=& -L_t^2 + k_{tt}\,,\\
\partial_t T &=& -C_t S^2+ h_t\,.
\end{eqnarray}
\end{prop}
where 
\begin{equation}
l_t= g_0 (g^{-1})^{z}_t l_z\,,\quad k_{tt}= g_0^2 (g^{-2})^{z}_t k_{zz}\,,\quad s_t=g_0 (g^{-1})^{z}_t s_z\,,\quad h_t= g_0^3 (g^{-1})^z_t h_z\,.
\end{equation}
\begin{proof}
These relations follow from the definitions and from Eq.~\ref{Leq}.
\end{proof}
\begin{rem}
We note that not all these generators are algebraically independent. We have for example:
\begin{equation}
L_t=C_t T\,.
\end{equation}
and furthermore, using $D_z C_z=0$ and defining $C_t:= g_0^{-1} (g^{-1})^{z}_t C_z$, we have
$\partial_t C_t=0$ and hence $C_t$ is a constant. We have thus a further algebraic relation between generators:
\begin{equation}
g_z^t= g_0^{-1} C_t^{-1} C_z\,.
\end{equation}
Out of the special polynomial ring generators $g_0, g_z^t,L_t$ and $T$ there are thus only two algebraically independent ones $g_0$ and $L_t$.
\end{rem}

\subsection{Calabi-Yau twofolds}\label{K3diffring}
We introduce $S^i \in \Gamma_{C^{\infty}}(\mathcal{L}^{-2} \otimes T \mathcal{M}),S \in \Gamma_{C^{\infty}}(\mathcal{L}^{-2})$ which satisfy:
\begin{equation}
\partial_{\ib} S = G_{i\ib} S^i \,,\quad \partial_{\bar{\imath}}  S^i = e^{2K} G^{i\jb}\ov{C}_{\jb\bar{\imath}},
\end{equation}
To define $S^i$ we integrate Eq.~(\ref{curvatureK3}) to obtain:
\begin{equation}\label{specgeomK3}
\Gamma_{im}^k=\delta^{k}_i K_{m} +\delta_m^k K_{i} +  C_{im} S^k + s_{im}^k \,,
\end{equation}

We then have the following: 

\begin{prop}
\begin{eqnarray}
D_i S^j &=& \delta_i^j S +C_{ik} S^j S^k + h_i^j\,, \\
D_i S &=& C_{ij} S^j S + K_j h^j_i +h_i \, ,\\
D_i K_j &=& -K_i K_j -C_{ij}S^k K_k + C_{ij} S + k_{ij}\,.
\end{eqnarray}
\end{prop}
\begin{proof}
The proof is purely computational and follows similar steps as in the case of the elliptic curve and for threefolds, where the analogous differential ring was put forward in Ref.~\cite{Alim:2007qj}.
\end{proof}

We now define the generators of the special differential ring:
\begin{dfn}The generators of the special differential ring are given by:
\begin{eqnarray}
g_0 &:=& \mathsf{h}_0^{-1}\,e^{-K}\,,\\
g_i^a &:=& e^{-K} G_{i\ib} \, \mathsf{h}^{\ib a}\,,\\
T^a &:=& g_0 \, g_i^a\, S^i\,,\\
T&:=& g^2_0 (S-S^iK_i)\,,\\
L_a &:=& g_0 (g^{-1})^i_a\, (K_i-l_i)\,,
\end{eqnarray}
\end{dfn}
Recalling the derivative:
\begin{equation}
\partial_a := \mathsf{h}_{a\ib} G^{i\ib} \partial_i\,,\quad \partial_a|_{\textrm{hol}}= \frac{\partial}{\partial t^a}\,,
\end{equation}
we obtain the following 
\begin{prop}{The special differential ring is given by:}
\begin{eqnarray}
\partial_a g_0 &=& - g_0\, L_a\,,\\
\partial_a \,g_i^b&=& \delta_a^b g_i^c L_c + g_i^c C_{ac} T^b + g_0\,  g_i^c s_{ac}^b \,, \\
 \partial_a T^b &=& \delta_a^b (T+ T^cL_c) + C_{ac} T^c T^b+ g_0^2\,h_a^b\,,\\
 \partial_a T&=& -T L_a + g_0^3 h_a\,,\\
 \partial_a L_b &=& -L_a L_b + C_{ab} T + g_0^2 k_{ab}\,.
 \end{eqnarray}
\end{prop}

where
\begin{eqnarray}
s_{ac}^b=  (g^{-1})^i_a (g^{-1})^j_a \,g^b_k s_{ij}^k\,,\quad h_a^b= (g^{-1})^i_a \, g^b_j \,h_i^j\,, \\
h_a= (g^{-1})^i_a h_i\,, \quad k_{ab}=(g^{-1})^i_a (g^{-1})^j_b k_{ij}\,,
\end{eqnarray}
are rational functions. 
 
\begin{rem} 
Similarly to the case of the elliptic curve, the generators of the differential ring are not algebraically independent. This can be seen from integrating Eq.~(\ref{K3algrel1}), which gives:
\begin{equation}
K_{m} = - C_{mj} S^j + k_m\,, 
\end{equation}
rewriting the last equation as:
\begin{equation}
G_{j\jb}(C^{-1})^{jm} (K_m-k_m)= -\partial_{\jb} S\,,
\end{equation}
we further deduce:
\begin{equation}
S=- \frac{1}{2} (C^{-1})^{mj} K_m K_j -  (C^{-1})^{mj} K_j k_m\,.
\end{equation}
For the generators of the special differential ring we have further relations, we first introduce:
 \begin{equation}
 C_{ab} = (g^{-1})^i_a\, (g^{-1})^j_b C_{ij}\,,
 \end{equation}
 and note that $D_i C_{jk}=0$ becomes:
 \begin{equation}
 \partial_a C_{bc}=0 \,, 
 \end{equation}
 so $C_{bc}$ is a constant matrix, and since $C_{ij}$ is a rational function in the moduli $z^i$ then $g_i^a$ will also be given by (square roots of) rational functions.
We furthermore translate to the special generators the algebraic relations:
 \begin{equation}
 T=\frac{1}{2} C^{ab} L_a L_b\,,
 \end{equation}
 and 
 \begin{equation}
 L_a= -C_{ab}T^b + g_0 k_a\,,
 \end{equation}
 where $k_a= (g^{-1})^i_a k_i$.
  \end{rem}
 
 \subsection{Calabi-Yau threefolds}
We introduce the objects $S^{ij},S^i,S$, the propagators or Ref.~\cite{Bershadsky:1993cx}, which are sections of $\mathcal{L}^{-2}\otimes \text{Sym}^m T\mathcal{M}$ with $m=2,1,0$, respectively, and give local potentials for the non-holomorphic Yukawa couplings:
\begin{equation}
\partial_{\bar{\imath}} S^{ij}=e^{2K} G^{j\jb}G^{i\ib} \overline{C}_{\bar{\imath}\jb\kb}, \qquad
\partial_{\bar{\imath}} S^j = G_{i\bar{\imath}} S^{ij}, \qquad
\partial_{\bar{\imath}} S = G_{i \bar{\imath}} S^i.
\label{prop}
\end{equation}
A differential ring in these objects and $K_i$ was put forward in Ref.~\cite{Alim:2007qj} and further enhanced to a special differential ring in Ref.~\cite{Alim:2014dea}. We will recall the latter in the following: 
\begin{dfn}\label{specgen}
The generators of the special polynomial differential ring are defined by
\begin{align}
g_0 &:=\mathsf{h}_0^{-1} e^{-K} \,, \\
 g^a_i& :=e^{-K} G_{i\jb}\, \mathsf{h}^{\jb a} \,,\\
T^{ab} & := g^a_i \, g^b_j S^{ij}  \,,\\
T^a & :=  g_0 \, g^a_i (S^i-S^{ij}K_j)\,,\\
T & := g_0^2 (S-S^i K_i + \frac{1}{2} S^{ij} K_i K_j)\,,\\
L_a &= g_0 (g^{-1})^{i}_a \partial_i K\,.
\end{align}
\end{dfn}

\begin{prop}{The special differential ring}\cite{Alim:2014dea}
\begin{align}
 \partial_a g_0 &= -L_a \, g_0 \,, \\ 
\partial_a g_i^b &= g_i^c\left( \delta^{b}_a\,L_c - C_{cad} T^{bd} + g_0\, s_{ca}^b\right)\,,\\
\partial_a T^{bc} & = \delta_{a}^{b} (T^c + T^{cd} L_d)+  \delta_{a}^{c} (T^b + T^{bd} L_d) - C_{ade} T^{bd} T^{ce} + g_0\, h_{a}^{bc}\, ,\\
\partial_a T^b &=  2 \delta^{b}_a (T + T^c L_c) -T^b\, L_a - k_{ac} T^{bc} + g_0^2\, h^b_a\, ,\\
\partial_a T & =  \frac{1}{2} C_{abc} T^b T^c - 2 L_a T- k_{ab} T^b + g_0^3\, h_a\, , \\
\partial_a L_b & =  -L_a L_b -C_{abc} T^c +  g_0^{-2}\, k_{ab}\,.
\end{align}
\end{prop}

\begin{proof}
See Refs.~\cite{Alim:2007qj,Alim:2014dea}.
\end{proof}
In these equations the functions $s_{ij}^k, h_i^{jk},h_i^j,h_i$ and $k_{ij}$ are fixed once a choice of generators has been made and we transformed the indices from arbitrary algebraic coordinates to the special coordinates using $g_i^a$ and its inverse. The freedom in choosing generators is discussed in Refs.~\cite{Alim:2008kp,Hosono:2008np,Alim:2014dea}.


 \section{Algebraic structure of $tt^*$ equations}\label{liealgebra}
The differential rings developed in the previous section are interesting on their own right, they are typically transcendental functions, when expressed in $\mathcal{M}$ and they coincide with quasi modular forms for special cases. For the elliptic curve and some lattice polarized $K3$ this will be shown in Sec.~\ref{examples}. For CY threefolds, the differential rings coincide with quasi modular forms when limits are considered in the moduli space \cite{Alim:2012ss,Alim:2013eja} but they give structures which generalize these for generic CY manifolds as advocated in Refs.~\cite{Yamaguchi:2004bt,Alim:2007qj,Mov11,Mov12a}. It was moreover proved in Refs.~\cite{Yamaguchi:2004bt,Alim:2007qj} that the higher genus topological string amplitudes of Ref.~\cite{Bershadsky:1993cx} can be expressed as polynomials in these generators. Expansions of these amplitudes in particular loci in $\mathcal{M}$ give in particular the generating functions of Gromov-Witten invariants. In this section it will be shown that the generators of the differential rings serve another purpose, namely that of providing an algebraic description of the $tt^*$ equations. To this end, we will put an idea of Movasati for the elliptic curve \cite{Mov12b,Mov12c} into the $tt^*$ context. In these works, it was shown that once the moduli space of an elliptic curve was enhanced to include a choice of differentials obeying the Hodge filtration while preserving a constant symplectic pairing that one is naturally led to a larger three dimensional space with the action of an algebraic group whose Lie algebra is $sl(2,\mathbbm{C})$. It was moreover shown that coordinates on the larger space correspond to the classical elliptic quasi modular forms. This approach was generalized to CY threefolds in Ref.~\cite{Alim:2014dea}, providing an algebraic reformulation of BCOV \cite{Bershadsky:1993cx}. In this section, the $tt^*$ context of these general ideas will be worked out, extending previous works further to include lattice polarized $K3$ manifolds. To this end, we will deduce the relation between two different choices of Hodge filtrations obtained by acting on the holomorphic top form using derivatives. In the $2d$ SCFT language this corresponds to different choices of representatives of the chiral ring elements in the deformation bundle obtained from the repeated action of the truly marginal operators on the unique NS ground state. It turns out, that at every step one has the freedom to add multiples of states which were already generated in the previous steps. In addition to this there are choices of coordinates and of scales which preserve the constant metric. The parameterization of all this freedom will be given by the generators of the differential rings. This procedure will be elucidated in the following.

 \subsection{Elliptic Curve}
 We will start with the elliptic curve and recover a result of Ref.~\cite{Mov12c} in a $tt^*$ context.
 We first note that from
 \begin{equation}
 D_z C_z =0 \,,
 \end{equation}
 it follows that:
 \begin{equation}
\partial_t C_t=0\,,\quad \textrm{for} \quad C_t := g_0^{-1} (g^{-1})^{z}_t C_z=c\,,
 \end{equation}
where $c$ is a constant which will be set to 1.

\subsubsection{Parameterizing $\mathsf{T}$}
We start with a choice of Hodge filtration which is obtained from the holomorphic $(1,0)$ form using derivatives with respect to arbitrary local coordinates $z$. This choice is given by:
\begin{equation}
\vec{\Omega}_z :=\left(\begin{array}{c}  \Omega \\ -C_z^{-1} \partial_z \Omega \end{array}\right) \in \begin{array}{l}  \mathcal{F}^1=H^{1,0}(\CY) \\ \mathcal{F}^0=H^{1,0}(\CY)\oplus H^{0,1}(\CY)=H^{1}(\CY,\mathbbm{C}) \end{array} 
\end{equation}

A distinguished choice of filtration is given by:
\begin{equation}
\vec{\Omega}_t :=\left(\begin{array}{c} \tilde{\Omega} \\ -\partial_{t}\tilde{\Omega} \end{array}\right) \in \begin{array}{l}  \mathcal{F}^1=H^{1,0}(\CY) \\ \mathcal{F}^0=H^{1,0}(\CY)\oplus H^{0,1}(\CY)=H^{1}(\CY,\mathbbm{C}) \end{array} \, ,
\end{equation}
where $\tilde{\Omega}=g_0^{-1}\Omega$ and $\partial_t := \mathsf{h}_{t\bar{z}},G^{z\bar{z}}\partial_z$ is chosen such that it gives the derivative with respect to the flat coordinate discussed in Sec.~\ref{ellflat}.
The choice $\vec{\Omega}_t$ is distinguished since:
\begin{equation}
\partial_t \vec{\Omega}_t = \left(\begin{array}{cc}  0 & C_t \\  0 & 0 \end{array}\right) \vec{\Omega}_t= \left(\begin{array}{cc}  0 & 1 \\  0 & 0 \end{array}\right) \vec{\Omega}_t\,.
\end{equation}
We compute using the differential ring of Sec.~\ref{diffringell} the matrix $\mathsf{B}$ relating the two different choices of Hodge filtration:
\begin{equation}\label{filtrel}
\vec{\Omega}_t = \mathsf{B} \cdot \vec{\Omega}_z\,.
\end{equation}
and find:
\begin{equation}
\mathsf{B}= \left( \begin{array}{cc} g_0^{-1}&0 \\ -g_0^{-1}(L_t + l_t) & g_0 \end{array}\right)\,.
\end{equation}

\subsubsection{Lie algebra}
We now consider the larger moduli space $\mathsf{T}$, which parameterizes the complex structure of the elliptic curve plus a choice of filtration.
\begin{equation}
\mathsf{T} = \textrm{Moduli space of} \left( \CY, \vec{\Omega} \,\, |\,\,  \langle \vec{\Omega},\vec{\Omega}\rangle = \Phi \right)\,,  \quad \Phi = \left(\begin{array}{cc} 0&1\\ -1 &0 \end{array}\right) \,.
\end{equation}
It is shown in Ref.~\cite{Mov12c} that this space is three-dimensional. The computation above shows that we can choose as local coordinates on this space $\{g_0,L_t,z\}$. We can now think of $\vec{\Omega}_t$ as living in this three-dimensional space and consider the action of vector fields $\partial_t,\partial_{g_{0}},\partial_{L_t}$ on it, i.e. from Eq.~\ref{filtrel}:
\begin{equation}
\partial_g \vec{\Omega}_t= M_{g} \cdot \vec{\Omega}_t = \partial_g \mathsf{B} \cdot \mathsf{B}^{-1} \vec{\Omega}_t + \mathsf{B} \,\partial_g \vec{\Omega}_z\,.
\end{equation}
which defines the matrices $M_{g}$
We compute the following:
\begin{equation}
M_t= \left( \begin{array}{cc} 0 &1 \\0 &0  \end{array}\right)\,,\quad M_{g_0}=g_0^{-1} \left( \begin{array}{cc} -1  & 0 \\ 2L_t & 1  \end{array}\right)
\,,\quad M_{L_t}= \left( \begin{array}{cc}  0&0 \\ -1& 0 \end{array}\right)\,.
\end{equation}
We now consider the following linear combinations:
\begin{eqnarray}
\mathcal{J}_+ &=& M_{t}= \left( \begin{array}{cc} 0 &1 \\0 &0  \end{array}\right)\,,\\
\mathcal{J}_{-}&=& - M_{L_t}=\left( \begin{array}{cc} 0 &0 \\1 &0  \end{array}\right)\,,\\
\mathcal{J}_0&=& -g_0 M_{g_0} + 2 M_{L_t}=\left( \begin{array}{cc} 1 &0 \\0 &-1  \end{array}\right)\,.
\end{eqnarray}
These clearly form a basis of an $sl(2,\mathbbm{C})$ Lie algebra with the commutators:
\begin{equation}
\left[\mathcal{J}_+,\mathcal{J}_- \right] = \mathcal{J}_0\,,\quad \left[ \mathcal{J}_0,\mathcal{J}_+\right]=2 \mathcal{J}_+\,,\quad \left[ \mathcal{J}_0,\mathcal{J}_-\right]=-2\mathcal{J}_-\,.
\end{equation}
This provides a proof of Theorem \ref{vectorfields} for the case $d=1$.
\subsubsection{Non-holomorphic content of $tt^*$}
 We want to elucidate in the following in which sense the $sl(2,\mathbbm{C})$ Lie algebra given above captures the non-holomorphic content of the $tt^*$ equations in this case. In the above derivation, the choice of elements in the filtration space given by $\vec{\Omega}_z$ are holomorphic. The entries of $\vec{\Omega}_t$ are non-holomorphic since they are constructed using the ingredients of the $tt^*$ connection, they become holomorphic after taking the holomorphic limit. We recall that:
 \begin{equation}
 \vec{\Omega}_t :=\left(\begin{array}{c} \tilde{\Omega} \\ -\partial_{t}\tilde{\Omega} \end{array}\right) = \left(\begin{array}{c} \mathsf{h}_0\, e^K\, \Omega \\ -  \mathsf{h}_{\bar{z}t}\,G^{z\bar{z}}\, \partial_{z} ( \mathsf{h}_0\, e^K\, \Omega)  \end{array}\right)\,,
 \end{equation}
We can compute explicitly:
\begin{equation}
\partial_t \vec{\Omega}_t=\mathsf{h}_{\bar{z}t}G^{z\bar{z}} \partial_z \vec{\Omega}_t = \left( \begin{array}{cc} 0 &1 \\0 &0  \end{array}\right) \vec{\Omega}_t\,.
\end{equation}
Since the entries are non-holomorphic, we can also compute:
\begin{equation}
\mathsf{h}^{\bar{z}t} \, D_{\bar{z}}\vec{\Omega}_t= \left( \begin{array}{cc} 0 &0 \\1 &0  \end{array}\right) \vec{\Omega}_t\,, 
\end{equation}
when the holomorphic limit is taken as discussed in Sec.~\ref{holomorphiclimits} a priori the derivative $D_{\bar{z}}$ does not make sense anymore. However it was shown above that the non-holomorphic derivative can be replaced by a derivative with respect to the generators of the special differential ring which parameterize $\mathsf{T}$. This notion of derivation survives the holomorphic limit and can be phrased in a purely holomorphic context.

 \subsection{Calabi-Yau twofolds}
 We proceed with a similar discussion for lattice polarized K3 manifolds. 

\subsubsection{Parameterizing $\mathsf{T}$}
We chose a basis $\vec{\Omega}=\{ \alpha_0,\alpha_1,\dots,\alpha_{n+1}\}$ of the part $H^2(\CY,\mathbbm{C})$ coming from the transcendental lattice $T(\CY)$ of $\CY$, where $n=\textrm{dim} \mathcal{M}$ such that:
\begin{equation}
\int_\CY \alpha_0 \wedge \alpha_{n+1} =-1\,\quad \int_\CY \alpha_{a} \wedge \alpha_b = C_{ab}\,,\quad \,, a,b=1,\dots,n,
\end{equation}
which can be written as:
\begin{equation}
\langle \vec{\Omega},\vec{\Omega}\rangle = \Phi = \left( \begin{array}{ccc} 0 & 0&-1 \\0&C_{ab} &0 \\-1&0&0 \end{array} \right)
\end{equation}
we furthermore require $\alpha_0 \in \mathcal{F}^2,\alpha_0,\alpha_1\,,\dots \alpha_n \in \mathcal{F}^1\, \alpha_0,\dots,\alpha_{n+1}\in\mathcal{F}^0$.
Let:
\begin{equation}
\mathsf{T}= \textrm{Moduli space of}\left\{ \CY,\vec{\Omega} \quad|\, \langle \vec{\Omega},\vec{\Omega}\rangle = \Phi \right\}\,,
\end{equation}
where the entries of $\vec{\Omega}$ are obeying the Hodge filtration as above.

We make a different choice of filtration which we call $\vec{\Omega}_z$ which we obtain by taking partial derivatives w.r.t the algebraic local coordinates $z^i$.
\begin{equation}
\vec{\Omega}_z := \left( \begin{array}{c}   \Omega \\
(g^{-1})^i_a \partial_i \Omega\\
(C^{-1})^{\sharp j} \partial_{\sharp} \partial_j \Omega 
\end{array}\right)\,,
\end{equation}

where $\partial_{\sharp}:=(g^{-1})^i_* \partial_i $ and $C_{\sharp j}:= (g^{-1})^i_* C_{ij}$, with $*$ a fixed modulus.

We further denote by:
\begin{equation}
\vec{\Omega}_t := \left( \begin{array}{c}   \tilde{\Omega} \\
\partial_a \tilde{\Omega}\\
(C^{-1})^{*b} \partial_* \partial_b \tilde{\Omega} 
\end{array}\right)
\end{equation}
where $*$ denotes a fixed modulus and where $\tilde{\Omega}=g_0^{-1} \Omega$ such that in the holomorphic limit $\tilde{\Omega}$ has an expansion:
\begin{equation}
\tilde{\Omega}:= g_0^{-1}\, \Omega= \alpha_0 + t^a \alpha_a +\frac{1}{2} C_{ab}t^a t^b \alpha_{n+1}\,,
\end{equation}

\begin{prop} We have the following:
\begin{enumerate}
\item
\begin{equation}
\int_{\gamma^{\mu}}   \left( \begin{array}{c}   \tilde{\Omega} \\
\partial_a \tilde{\Omega}\\
(C^{-1})^{*d} \partial_* \partial_d \tilde{\Omega} 
\end{array}\right) = \left( \begin{array}{ccc} 1 &t^b &\frac{1}{2}C_{cd}t^c t^d\\
0&\delta_a^b&C_{ac}t^c\\
0&0&1
\end{array}\right)\,,
\end{equation}

\item The intersection matrix is given by: 
\begin{equation}
\langle \vec{\alpha}_t,\vec{\alpha}_t\rangle = \Phi = \left( \begin{array}{ccc} 0 & 0&-1 \\0&C_{ab} &0 \\-1&0&0 \end{array} \right)
\end{equation}

\item
\begin{equation}
\partial_c   \left( \begin{array}{c}   \tilde{\Omega} \\
\partial_a \tilde{\Omega}\\
(C^{-1})^{*d} \partial_* \partial_d \tilde{\Omega} 
\end{array}\right) = \left( \begin{array}{ccc} 0 &\delta_c^b &0\\
0&0&C_{ac}\\
0&0&0
\end{array}\right) \, \left( \begin{array}{c}   \tilde{\Omega} \\
\partial_b \tilde{\Omega}\\
(C^{-1})^{*d} \partial_* \partial_d \tilde{\Omega} 
\end{array}\right)\,,
\end{equation}
 
\end{enumerate}
\end{prop}
\begin{proof}
These properties follow by construction.
\end{proof}
We now compute the matrix $\mathsf{B}$ which relates the two choices of filtrations:
\begin{equation}
\alpha_t = \mathsf{B}\, \alpha_z\,.
\end{equation}
We find the following, after using the algebraic relations between the generators:
\begin{equation}
\mathsf{B} = \left( \begin{array}{ccc}
g_0^{-1} &0&0\\
g_0^{-1}L_a & \delta_a^b &0 \\
g_0^{-1}(\frac{1}{2} C^{de} L_d L_e+g_0^2\, \mathcal{E}) &  (C^{bd}L_d - g_0\, \mathcal{E}^b) & g_0
\end{array}\right) \,,
\end{equation}
with $\mathcal{E}=(C^{-1})^{*b} k_{*b}$ and $\mathcal{E}^b= C^{*c} s_{*c}^b+ C^{bc} k_c$.

\subsubsection{Lie algebra}
We now define matrices $M_g$, which are obtained by considering the generators of the special differential ring as independent coordinates on $\mathsf{T}$ and by considering $\vec{\Omega}_t$ as a function of these:
\begin{equation}
\partial_g \vec{\Omega}_t= M_g \vec{\Omega}_t\,.
\end{equation}
The matrices $M_g$ for $g\in \{g_0,L_a\}$ are given in the appendix. We now consider the linear combinations:
\begin{eqnarray}
\mathcal{J}_0 = -g_0 M_{g_0} - M_{K_m} K_m = \left( \begin{array}{ccc}
1 &0&0\\
 0& 0 &0 \\
0&  0& -1
\end{array}\right) \,,  \\
\mathcal{J}^+_m= M_{\partial_m}= \left( \begin{array}{ccc}
0 &\delta_m^b&0\\
 0& 0 &C_{am} \\
0&  0& 0
\end{array}\right) \,,\\
(\mathcal{J}_-)_n= C_{nl} M_{K_l}= \left( \begin{array}{ccc}
0 &0&0\\
 C_{na}& 0 &0 \\
0&  \delta_n^b& 0
\end{array}\right) \,,
\end{eqnarray}

we can now compute the commutators of these generators:
\begin{eqnarray}
\left[ \mathcal{J}^+_m,\mathcal{J}_{-,n}\right] &=&- C_{mn} \mathcal{J}_0 +\mathcal{Z}_{mn}\,, \\
\left[ \mathcal{J}_0,\mathcal{J}^+_m\right] &=&  \mathcal{J}^+_m\,,\\
\left[ \mathcal{J}_0,\mathcal{J}_{-,m}\right] &=& - (\mathcal{J}_-)_m\,,
\end{eqnarray}

with:
\begin{equation}
\mathcal{Z}_{mn} = \left( \begin{array}{ccc}
0 &0&0\\
 0& C_{am}\delta^c_n-C_{na}\delta_m^c &0 \\
0&  0& 0
\end{array}\right) \,.  \\
\end{equation}

This provides a proof of Theorem \ref{vectorfields} for the case $d=2$. A full discussion and classification of the Lie algebras which can appear crucially depends on $C_{ab}$. In particular the appearance of the term $\mathcal{Z}_{mn}$ should be absorbed intro a redefinition of the generators. We expect that the classification of the Lie algebras which appear here will be initmately related to the $ADE$ type classification of singularities in K3 moduli spaces, this is however beyond the scope of this work. For one dimensional moduli spaces it is however possible to identify the Lie algebra.

\subsubsection{One dimensional moduli space}
We note that in the one dimensional moduli space this Lie algebra becomes simple.
The coordinates on the larger moduli space $\mathsf{T}$ are  $z,g_0,L_t$. The Lie algebra is given by three generators:
\begin{eqnarray}
\mathcal{J}_0 = -g_0 M_{g_0} - M_{L_t} L_t = \left( \begin{array}{ccc}
1 &0&0\\
 0& 0 &0 \\
0&  0& -1
\end{array}\right) \,,  \\
\mathcal{J}^+= M_{\partial_t}= \left( \begin{array}{ccc}
0 &1&0\\
 0& 0 &1 \\
0&  0& 0
\end{array}\right) \,,\\
(\mathcal{J}_-)=  M_{L_t}= \left( \begin{array}{ccc}
0 &0&0\\
 1& 0 &0 \\
0&  1& 0
\end{array}\right) \,,
\end{eqnarray}
which form a basis of an $sl(2,\mathbbm{C})$ Lie algebra with the commutators:
\begin{equation}
\left[\mathcal{J}_+,\mathcal{J}_- \right] = \mathcal{J}_0\,,\quad \left[ \mathcal{J}_0,\mathcal{J}_+\right]=\mathcal{J}_+\,,\quad \left[ \mathcal{J}_0,\mathcal{J}_-\right]=-\mathcal{J}_-\,.
\end{equation}
which can be transformed to the more standard form:
\begin{equation}
\left[\mathcal{\tilde{J}}_+,\mathcal{\tilde{J}}_- \right] = \mathcal{\tilde{J}}_0\,,\quad \left[ \mathcal{\tilde{J}}_0,\mathcal{\tilde{J}}_+\right]=2 \mathcal{\tilde{J}}_+\,,\quad \left[ \mathcal{\tilde{J}}_0,\mathcal{\tilde{J}}_-\right]=-2\mathcal{\tilde{J}}_-\,.
\end{equation}
with the redefinitions:
\begin{equation}
\tilde{\mathcal{J}}_0 = 2 \mathcal{J}\,,\quad \tilde{\mathcal{J}}_{\pm}=\sqrt{2} \mathcal{J}_{\pm}\,.
\end{equation}


\subsection{Calabi-Yau threefolds}
The discussion for CY threefolds follows Ref.~\cite{Alim:2014dea}. We follow similar steps as in the elliptic curve and lattice polarized $K3$ cases and define two different choices of Hodge filtrations having constant pairing.

The choices are given by
\begin{equation}
\vec{\Omega}_z = \left(\begin{array}{c} \alpha_{z,0} \\ \alpha_{z,i} \\ \beta_z^i \\ \beta_z^0\end{array} \right)\left( \begin{array}{c} \Omega  \\ \partial_i \Omega \\ (C_{\sharp}^{-1})^{ik} \partial_{\sharp} \partial_k \Omega \\ -\partial_{\sharp} (C_{\sharp}^{-1})^{\sharp k} \partial_* \partial_k \Omega \end{array}\right),
\end{equation}
here, $A_{\sharp}= (g^{-1})^{i}_* A_i$, where $*$ denotes a fixed choice of special coordinate.
We also define
\begin{equation}
\label{6oct2014}
\vec{\Omega}_t= \left(\begin{array}{c} \alpha_{t,0} \\ \alpha_{t,a} \\ \beta_t^a \\ \beta_t^0\end{array} \right)=
\left(\begin{array}{c}
\tilde{\Omega}\\
\partial_a \tilde{\Omega} \\
(C_*^{-1})^{ae}\partial_* \partial_e \tilde{\Omega}\\
-\partial_* (C_*^{-1})^{*e}\partial_* \partial_e \tilde{\Omega}
\end{array}
\right)\,,
\end{equation}
where $\tilde\Omega=\tilde\Omega_t$ is given by \eqref{cmsa2014} and
$*$ denotes a fixed choice of special coordinate.

\begin{prop}
\begin{enumerate}
\item
The period matrix of $\vec{\Omega}_t$ over
the symplectic basis of $H_3(\CY_z,\Z)$  has the format:
\begin{equation}
\label{specialmatrix}
[\int_{{A^0,A^c,B_c,B_0} }\left(\begin{array}{c} \alpha_0 \\ \alpha_a \\ \beta^a \\ \beta^0\end{array} \right)]= \left(\begin{array}{cccc}
1 & t^c & F_c & 2F_0-t^dF_d\\
0&\delta^c_a & F_{ac}& F_a-t^d F_{ad}\\
0&0&\delta^a_c & -t^a \\
0&0&0&1
\end{array}
\right)\,,
\end{equation}
with $F_a:=\partial_a F_0,\quad \partial_a=\frac{\partial}{\partial t^a}.$
\item 
The symplectic form for both bases $\vec{\Omega}_z$ and $\vec{\Omega}_{t}$ is the matrix $\imc$ given by:
\begin{equation}
\imc = \left(\begin{array}{cccc} 0&0&0&1 \\ 0&0&\mathbbm{1}_{n\times n}&0 \\ 0&-\mathbbm{1}_{n\times n}&0&0\\ -1&0&0&0 \   \end{array} \right)
\end{equation}
\item 
The flat choice $ \vec{\Omega}_{t} $ satisfies the following equation:

\begin{equation}
\label{determinantal}
\partial_b \left(\begin{array}{c}
\tilde{\Omega}\\
\partial_a \tilde{\Omega} \\
(C_*^{-1})^{ae}\partial_* \partial_e \tilde{\Omega}\\
-\partial_* (C_*^{-1})^{*e}\partial_* \partial_e \tilde{\Omega}
\end{array}
\right) =   \left(\begin{array}{cccc}
0& \delta_b^c& 0& 0\\
0&0 & C_{abc}&0\\
0&0&0& -\delta^a_b \\
0&0&0&0
\end{array}
\right) \left(\begin{array}{c}
\tilde{\Omega}\\
\partial_c \tilde{\Omega} \\
(C_*^{-1})^{ce}\partial_* \partial_e \tilde{\Omega}\\
-\partial_* (C_*^{-1})^{*e}\partial_* \partial_e \tilde{\Omega}
\end{array}
\right)\,.
\end{equation}

\end{enumerate}

\end{prop}

\begin{proof}
 All of these follow from the definitions
  \end{proof}

\def\Am{{\sf B}}
We now want to find the matrix relating:
\begin{equation}
 \vec{\Omega}_{t}  = \Am\cdot \vec{\Omega_z}\,,
\end{equation}
and express its entries in terms of the polynomial generators.
The matrix $\Am$ is given in the appendix.


\subsubsection{Lie Algebra description}

In the following, we regard all the generators as independent variables and 
compute the following matrices $M_g$ defined by:
\begin{equation}
\partial_g \vec{\Omega}_t = M_g \vec{\Omega}\,,
\end{equation}
where $g$ refers to coordinates on the larger moduli space $\mathsf{T}$. These coordinates correspond to the generators of the differential ring.
\begin{equation}
M_{g}= \frac{\partial}{\partial g} \Am \cdot \Am^{-1}\,,
\end{equation}
where $g$ denotes a generator and $\partial_g := \frac{\partial}{\partial g}$. 
 We now look for combinations of the vector fields which give constant vector fields. We find the following:

\begin{eqnarray}
\tgtwo_{ab} &=&  \mathsf{M}_{T^{ab}} -\frac{1}{2} (L_a\, \mathsf{M}_{T^b} + L_b \mathsf{M}_{T^a}) +\frac{1}{2} L_a L_b \mathsf{M}_{T}   \,,\\ \nonumber
\tgone_{a}&=&  \mathsf{M}_{T^a} - L_a \mathsf{M}_{T}\,, \\ \nonumber
\tgzero &=& \frac{1}{2}\mathsf{M}_{T}\,, \\ \nonumber
\kgone^{a}&=&   \mathsf{M}_{L_a}\,,\\ \nonumber
\ggtwo_{b}^{a}  &=& g^{a}_m \mathsf{M}_{g_{m}^{b}} - L_{b} \mathsf{M}_{L_a}  + 2 T^{ad} \mathsf{M}_{T^{db}} + T^{a} \mathsf{M}_{T^b}  \,, \\ \nonumber
\ggzero_0 &=& g_0 \mathsf{M}_{g_0} + L_{a} \mathsf{M}_{L_a} +T^{a} \mathsf{M}_{T^a} + 2 \mathsf{M}_{T}\,.
\end{eqnarray}
These combinations give the following generators of a Lie algebra \cite{Alim:2014dea}
\begin{equation}
\label{gofLie}
\tgtwo_{ab}:=\left( 
\begin{array}{cccc}
 0 & 0 & 0 & 0 \\
 0 & 0 & 0 & 0 \\
 0& \frac{1}{2} (\delta^{i}_{a} \delta^{j}_{b} + \delta^{i}_{b} \delta^{j}_{a} ) & 0 & 0 \\
 0 & 0 & 0 & 0 \\
\end{array}
\right) \
\tgone_{a}=
\left(
\begin{array}{cccc}
 0 & 0 & 0 & 0 \\
 0 & 0 & 0 & 0 \\
 -\delta^{i}_{a} & 0 & 0 & 0 \\
 0 & -\delta_{a}^{j} & 0 & 0 \\
\end{array}
\right)\,
\tgzero:=\left(
\begin{array}{cccc}
 0 & 0 & 0 & 0 \\
 0 & 0 & 0 & 0 \\
 0 & 0 & 0 & 0 \\
 1 & 0 & 0 & 0 \\
\end{array}
\right)\,
\end{equation}
$$
\kgone_{a}:= \left(
\begin{array}{cccc}
 0 & 0 & 0 & 0 \\
 \delta^{a}_{i} & 0 & 0 & 0 \\
 0 & 0 & 0 & 0 \\
 0 & 0 &- \delta^{a}_{j} & 0 \\
\end{array}
\right)\,,  \ \
\ggtwo^{a}_{b} :=
 \left(
\begin{array}{cccc}
 0 & 0 & 0 & 0 \\
 0 & -\delta_{i}^{a} \delta^{j}_{b}   & 0 & 0 \\
 0 & 0 & \delta_{b}^{i} \delta_{j}^{a} & 0 \\
 0 & 0 & 0& 0 \\
\end{array}\right)\,
\ggzero_0:=
\left(
\begin{array}{cccc}
 -1 & 0 & 0 & 0 \\
 0 & 0 & 0 & 0 \\
 0 & 0 & 0 & 0 \\
 0 & 0 & 0& 1 \\
\end{array}\right)\,.
$$

This provides a proof of Theorem \ref{vectorfields} for the case $d=3$ \cite{Alim:2014dea}.
The commutators are given in the following table :

\newpage
\begin{sideways}

 $
\begin{array}{|c|c|c|c|c|c|c|c|}
\hline
&\Ra_{\ggzero_0} & \Ra_{\ggtwo_{c}^{d}}&\Ra_{\tgtwo_{cd}} & \Ra_{\tgone_{c}} &  \Ra_{\tgzero} &\Ra_{\kgone^{c}}&\Ra_{c} \\

\hline
 \Ra_{\ggzero_0} &0&0&0&-\Ra_{\tgone_{c}}&-2\Ra_{\tgzero}&-\Ra_{\kgone^{c}} &\Ra_{c}\\

\hline
\Ra_{\ggtwo_{b}^{a}}&0&0&-\delta_{c}^{a} \Ra_{\tgtwo_{bd}}- \delta_{d}^{a} \Ra_{\tgtwo_{bc}}&- \delta^{a}_{c} \Ra_{\tgone_{b}}&0&\delta^{c}_{b} \Ra_{\kgone^{a}} &-\delta^{a}_{c} \Ra_{b}\\

\hline
\Ra_{\tgtwo_{ab}} &0&\delta_{a}^{d} \Ra_{\tgtwo_{bc}}+\delta_{b}^{d} \Ra_{\tgtwo_{ac}}&0&0&0&\frac{1}{2} (\delta_{a}^{c} \Ra_{\tgone_{b}} + \delta_{b}^{c} \Ra_{\tgone_{a}} )&-\frac{1}{2} ( \Yuk_{cbd} \Ra_{\ggtwo_{a}^{d}}  +\Yuk_{acd} \Ra_{\ggtwo_{b}^{d}}    )\\

\hline 
\Ra_{\tgone_{a}} &\Ra_{\tgone_{a}} &\delta^{d}_{a} \Ra_{\tgone_{c}}&0&0&0&2 \delta^{c}_{a}  \Ra_{\tgzero}&2 \Ra_{\tgtwo_{ac}} -\Yuk_{acd} \Ra_{\kgone^{d}} \\

\hline 
\Ra_{\tgzero} &2 \Ra_{\tgzero} &0&0&0&0&0&\Ra_{\tgone_{c}}\\

\hline
\Ra_{\kgone^{a}} &\Ra_{\kgone^{a}}&-\delta^{a}_{c} \Ra_{\kgone^{d}}&-\frac{1}{2} (\delta_{c}^{a} \Ra_{\tgone_{d}} + \delta_{d}^{a} \Ra_{\tgone_{c}} )& -2\delta^{a}_{c} \Ra_{\tgzero}&0&0&-\delta_{c}^{a} \Ra_{\ggzero_0} +\Ra_{\ggtwo_{c}^{a}}\\

\hline
\Ra_{a}&-\Ra_{a}&\delta_{a}^{d} \Ra_{c}&
\frac{1}{2} ( \Yuk_{ade} \Ra_{\ggtwo_{c}^{e}}  +\Yuk_{ace} \Ra_{\ggtwo_{d}^{e}}    )&
-2 \Ra_{\tgtwo_{ac}} +\Yuk_{ace} \Ra_{\kgone^{e}} &
-\Ra_{\tgone_{a}}&
\delta_{a}^{c} \Ra_{\ggzero_0} -\Ra_{\ggtwo^{c}_{a}}&
0\\
\hline
\end{array}
$

\end{sideways}


\section{Examples}\label{examples}
We will discuss some simple examples of CY, 1, 2 and 3 folds. The mirror geometries, which we consider in the following $\check{\CY}$ and $\CY$, are given by dual polyhedra $\Delta,\Delta^*$ using Batyrev's construction \cite{Batyrev:1994hm}, equivalently using Hori-Vafa's construction \cite{Hori:2000kt}. To avoid lengthy reviews of these techniques we refer to Refs.~\cite{Cox:1999,Hosono:2000eb,Alim:2012gq,Clader:2014kfa} and references therein for further background.

\subsection{Elliptic curve}

We consider the elliptic curve $\check{\CY}$ given by a section of the anti-canonical bundle over the weighted projective plane with weights $\mathbbm{P}_{1,2,3}$ and its mirror $\CY$. This can be described by the toric charge vector:
\begin{equation}
\begin{array}{c|ccc}
-6&3&2&1
\end{array}
\end{equation}
we define a local coordinate $z=\frac{a_1^3a_2^2 a_3}{a_0^6}$ on the moduli space $\mathcal{M}$ of $\CY$ and obtain the Picard-Fuchs operator:
\begin{equation}\label{PFcurve}
\mathcal{L}= \theta^2- 12 z (6\theta +5) (6 \theta+1)\,,\quad \theta=z\frac{d}{dz}\,,
\end{equation}
we further compute the discriminant of this operator:
\begin{equation}
\Delta= 1-432 z\,.
\end{equation}
Solutions of Eq.~\ref{PFcurve} are given by:
\begin{equation}
\omega_0= {}_2F_1(1/6/5/6,1,432 z)\,,\quad \omega_1=-2\pi {}_2F_1(1/6/5/6,1,1-432 z)\,,
\end{equation}
where we chose the normalization such that:
\begin{equation}
t =\frac{\omega_1}{\omega_0} = \log z +312 z + 58932 z^2 + \dots
\end{equation}

\subsubsection{Yukawa coupling}
The Yukawa coupling in this case is defined by:
\begin{equation}
C_z := -\int_\CY \Omega \wedge \partial_z \Omega \,,
\end{equation}
We can now compute
\begin{equation}
\partial_z C_z = -\int \partial_z \Omega \wedge \partial_z \Omega -\int \Omega \wedge \partial^2_z \Omega=\frac{1-2\Delta}{z\Delta} C_z\,,
\end{equation}
which is solved by
\begin{equation}
C_z = \frac{c}{z \Delta}\,,
\end{equation}
where $c$ is a constant which we will set to 1.

\subsubsection{Differential ring and quasi modular forms}
From the equations:
\begin{eqnarray}
\Gamma_{zz}^z= 2 L_z + s_z\,, \quad L_z+l_z=K_z\,,\quad D_z L_z =-L_z^2 + k_{zz}\,,
\end{eqnarray}
discussed in Sec.~\ref{diffringell} we make a choice of $l_z$ and compute $s_z,k_{zz}$ accordingly, this fixes the special differential ring. We use the holomorphic limit of $K_z$ which is given by:
\begin{equation}
K_z |_{hol} = -\partial_z \log \omega_0\,
\end{equation}
we choose:
\begin{equation}
l_z= -\frac{36}{1-432 z} +\frac{1}{12z}\,
\end{equation}
and compute:
\begin{equation}
s_z=\frac{360}{\Delta} -\frac{5}{6z}\,\quad k_{zz}= \frac{1}{(12 z \Delta)^2}\,.
\end{equation}

With this choice we have the generators of the special differential ring $g_0,L_t$ given by:
\begin{equation}
g_0(t)= \omega_0(\tau)=E_4^{1/4}(t)\,\quad L_t(t)=-\frac{E_2(t)}{12}\,,
\end{equation}
with the Eisenstein series $E_2$ and $E_4$.

\subsection{Lattice polarized K3}

We consider the $K3$ surface $\check{\CY}$ given by a quartic hypersurface in $\mathbbm{P}^3$, and its mirror $\CY$. This can be described by the toric charge vector:
\begin{equation}
\begin{array}{c|cccc}
-4&1&1&1&1
\end{array}
\end{equation}
we define a local coordinate $z=\frac{a_1 a_2 a_3 a_4}{a_0^4}$ on the moduli space $\mathcal{M}$ of the mirror quartic $\CY$ and obtain the Picard-Fuchs operator:
\begin{equation}\label{PFK3}
\mathcal{L}= \theta^3- 4 z \prod_{i=1}^3\,(4\theta +i) \,,\quad \theta=z\frac{d}{dz}\,,
\end{equation}
we further compute the discriminant of this operator:
\begin{equation}
\Delta= 1-256 z\,.
\end{equation}
The solutions of the PF equation are given by:
\begin{eqnarray}
\omega_0 &=&{}_3 F_2\left( \frac{1}{4},\frac{1}{2},\frac{3}{4};1,1,256 z\right)\,,\\
\omega_1&=& -\frac{1}{\sqrt{\pi} \Gamma(1/4)\Gamma(3/4)}\,G^{2\,3}_{3\,3} \left(  \begin{array}{ccc|} \frac{1}{4} &\frac{1}{2} & \frac{3}{4}\\ 0&0&0 \end{array}  \,\,256z \right)\,,\\
\omega_2&=& \omega_1^2/\omega_0\,.
\end{eqnarray}

We define\footnote{We absorbed a factor of $2\pi i$ in the definition.}:
\begin{equation}
t=\frac{\omega_1}{\omega_0}\, , \quad q=e^t\,,
\end{equation}
and obtain as the inverse mirror map:
\begin{equation}
z(q)=q-104 q^2+6444 q^3-311744 q^4+13018830 q^5+\mathcal{O}\left(q^6\right)
\end{equation}
The integrality of this mirror map has been addressed in Ref.~\cite{Lian:1994zv}.

\subsubsection{Yukawa coupling}
The Yukawa coupling in this case is defined by:
\begin{equation}
C_{zz} := -\int_\CY \Omega \wedge \partial_z\partial_z \Omega \,,
\end{equation}
We can now compute
\begin{equation}
\partial_z C_{zz} = -\int \partial_z \Omega \wedge \partial_{zz} \Omega -\int \Omega \wedge \partial^3_z \Omega=-\frac{2}{3}\frac{3-1152 z}{z\Delta} C_{zz}\,,
\end{equation}
which is solved by
\begin{equation}
C_z = \frac{c}{z^2 \Delta}\,,
\end{equation}
where $c$ is a constant which we will set to 1.

\subsubsection{Differential ring and modularity}
In the following we find choices of generators of the differential ring which will be expressed in terms of the differential ring of quasi modular forms associated to $\Gamma_0(2)$ subgroup of $SL(2,\mathbbm{Z})$.\footnote{A careful study of the PF operator and of the monodromy of its solutions reveals that the monodromy group in this case is $\Gamma_0(2)_{+}$ \cite{Lian:1994zv,Hosono:2000eb}} We 
will use the quasi modular forms, reviewed in appendix~\ref{modularappendix}:
\begin{eqnarray}
A(\tau) &=& (\theta_2(\tau) + \theta_3(\tau))^{1/2}\,,\\
B(\tau) &=& \theta_{4}^{2}(2\tau)\, \\
C(\tau) &=& \frac{1}{\sqrt{2}} \theta_{2}^{2}(\tau) \, \\
E(\tau) &=&\frac{1}{3}\left(2 E_2(2\tau) + E_2(\tau) \right) \,.
\end{eqnarray}

We start with an expression for the inverse mirror map, which we find to be:
\begin{equation}
z(t)= \frac{(C(t)^4 (A(t)^4 - C(t)^4))}{64 A(t)^8}\,, \quad \theta t(t)= \frac{C^2(t)}{A(t)^4-2C^4(t)}\,,
\end{equation}
we further find:
\begin{equation}
\omega_0(t) = A^2(t)\,.
\end{equation}

From the discussion in Sec.~\ref{K3diffring}, it is clear that in this case where the moduli space is only one-dimensional, the special differential ring has only 2 generators, which are $g_0$ and $L_t$. From $K_z=L_z+l_z$, we choose $l_z=\frac{1}{4z}$, which gives:
\begin{equation}
g_0(t) = A^2(t)\,\quad L_t(t) =  -\frac{E(t)}{4}\,.
\end{equation}

The algebraic constraints for the other generators are given by:
\begin{eqnarray}
1=C_{tt}= (g_z^t)^{-2} \frac{1}{z^2 \Delta}\, \quad g_z^t = \frac{1}{z \sqrt{\Delta}}\,,\\
T^t=-L_t\,,\quad T=\frac{1}{2} L_t^2\,.
\end{eqnarray}
The nontrivial differential ring relations are now:
\begin{eqnarray}
\partial_t g_0 &=& - (L_t-l_t) g_0 \\
\partial_t L_t &=& - \frac{1}{2} L_t^2 + g_0^2 k_{tt}\,.
\end{eqnarray}
The coordinates on the larger moduli space $\mathsf{T}$ are now $z,g_0,L_t$. The Lie algebra acting on $\mathsf{T}$ is given by three generators:
\begin{eqnarray}
\mathcal{J}_0 = -g_0 M_{g_0} - M_{L_t} L_t = \left( \begin{array}{ccc}
1 &0&0\\
 0& 0 &0 \\
0&  0& -1
\end{array}\right) \,,  \\
\mathcal{J}^+= M_{\partial_t}= \left( \begin{array}{ccc}
0 &1&0\\
 0& 0 &1 \\
0&  0& 0
\end{array}\right) \,,\\
(\mathcal{J}_-)=  M_{L_t}= \left( \begin{array}{ccc}
0 &0&0\\
 1& 0 &0 \\
0&  1& 0
\end{array}\right) \,,
\end{eqnarray}

which form a basis of an $sl(2,\mathbbm{C})$ Lie algebra with the commutators:
\begin{equation}
\left[\mathcal{J}_+,\mathcal{J}_- \right] = \mathcal{J}_0\,,\quad \left[ \mathcal{J}_0,\mathcal{J}_+\right]=\mathcal{J}_+\,,\quad \left[ \mathcal{J}_0,\mathcal{J}_-\right]=-\mathcal{J}_-\,.
\end{equation}
which can be transformed to the more standard form:
\begin{equation}
\left[\mathcal{\tilde{J}}_+,\mathcal{\tilde{J}}_- \right] = \mathcal{\tilde{J}}_0\,,\quad \left[ \mathcal{\tilde{J}}_0,\mathcal{\tilde{J}}_+\right]=2 \mathcal{\tilde{J}}_+\,,\quad \left[ \mathcal{\tilde{J}}_0,\mathcal{\tilde{J}}_-\right]=-2\mathcal{\tilde{J}}_-\,.
\end{equation}
with the redefinitions:
\begin{equation}
\tilde{\mathcal{J}}_0 = 2 \mathcal{J}\,,\quad \tilde{\mathcal{J}}_{\pm}=\sqrt{2} \mathcal{J}_{\pm}\,.
\end{equation}

\subsection{CY threefolds}

We consider the quintic threefold $\check{\CY}$ given by a quintic hypersurface in $\mathbbm{P}^4$, and its mirror $\CY$.
This can be described by the toric charge vector:
\begin{equation}
\begin{array}{c|ccccc}
-5&1&1&1&1&1
\end{array}
\end{equation}
we define a local coordinate $z$ on the moduli space $\mathcal{M}$ of the mirror $\CY$ and obtain the Picard-Fuchs operator:
\begin{equation}\label{PFquintic}
\mathcal{L}=\theta^4- 5z \prod_{i=1}^4 (5\theta+i)\,, \quad \theta=z \frac{d}{dz}\, .
\end{equation}
The discriminant of this operator is
\begin{equation}
  \label{eq:Discriminant}
\Delta=1-3125\,z\,.
\end{equation}
and the Yukawa coupling can be computed:
\begin{equation}
C_{zzz}=\frac{5}{z^3\, \Delta}\,.
\end{equation}

\subsubsection{Differential ring}
The choice of the differential ring generators is fixed such that the holomorphic functions appearing in the setup are rational expressions in terms of $z$, the coordinate in the large complex structure patch of the moduli space. For the holomorphic functions in the following lower/upper indices are multiplied/divided by $z$  
$$ A_i^j \rightarrow \frac{z_i}{z_j} A_i^j\,.$$
With this convention, all the holomorphic functions appearing in the setup of the polynomial construction can be expressed in terms of polynomials in the local coordinates (See Refs.~\cite{Yamaguchi:2004bt,Alim:2007qj,Hosono:2008np}). The holomorphic functions are chosen as in Ref.~\cite{Hosono:2008np,Alim:2012gq}:
\begin{equation}
s_{zz}^z=-\frac{8}{5}\,, \quad h_{z}^{zz}=\frac{1}{25}\,,\quad h^z_z=-\frac{1}{125}\,,\quad h_z=\frac{2}{3125}\,,\quad h_{zz}=\frac{2}{25}\,.
\end{equation}

\section{Conclusions and Discussion}
In this work we have shown that the non-holomorphic content of the $tt^*$ equations for sigma models into CY d-folds (d=1,2,3) can be phrased holomorphically and algebraically when working on a larger moduli space $\mathsf{T}$. In the $2d$ SCFT language the larger moduli space parameterizes different choices of representatives of the deformation bundle over the moduli space of the theory, respecting a topological pairing and a charge filtration. By considering vector fields on the larger moduli space, we obtained special combinations which could be understood as generators of a Lie algebra. The $tt^*$ equations for realizations of the $\mathcal{N}=(2,2)$ SCFT as nonlinear sigma models intro CY d-folds could then be phrased entirely algebraically.

Physically, this work paves the way to a better understanding of the $tt^*$ equations and their solutions. Although the class of geometries considered here correspond to the conformal ones, this class already includes the geometries which are relevant to geometrically engineer $4d, \mathcal{N}=2$ theories. The moduli spaces $\mathcal{M}$ considered in this work correspond to the Coulomb branch of these theories. The larger moduli spaces $\mathsf{T}$ and the algebraic characterization of the $tt^*$ equations in this case suggests that it is possible to phrase what a non-holomorphic deformation would mean in terms of deformations of purely holomorphic extra data, which are given by the choices of differentials.

Using the analogs of special geometry for higher dimensional CY manifolds \cite{Greene:1993vm}, it should be possible to work out the analogs of the differential rings and of the algebraic description of this work. The results provided here and relations of solutions of $tt^*$ equations to classical problems of mathematical physics \cite{Cecotti:1991me,Cecotti:1991vb,Cecotti:1992vy} suggest furthermore that it should be possible to obtain an algebraic and purely holomorphic reformulation of $tt^*$ equations in the general case.

Moreover, it was shown that coordinates on the larger moduli space are intrinsically interesting since these coincide with known quasi modular forms for some cases and provide their analogs more generally. For the case of lattice polarized $K3$ manifolds, it would be interesting to connect the Lie algebras which appear here to ADE type classification of singularities which occur in the moduli space, perhaps also a connection to recent developments in moonshine associated to K3 surfaces, see e.~g.~Ref.~\cite{Cheng:2014zpa} can be established. 

A set of functions which are naturally expressed in the larger moduli space are the topological string amplitudes for CY threefolds, which become polynomial in the differential ring generators \cite{Yamaguchi:2004bt,Alim:2007qj}. This was addressed in this algebraic context in Ref.~\cite{Alim:2014dea}. These amplitudes provide generating functions of GW invariants when expanded in certain loci in the moduli space. It is furthermore known that generating functions of GW invariants for the elliptic curve and for K3 surfaces are also expressed in terms of quasi modular forms, See Refs.~\cite{Dijkgraaf:1995,Pandharipande:2014ooa} and references therein. The tools developed in this work should pave the way towards a systematic understanding of the analogs of the recursive holomorphic anomaly equations of threefolds \cite{Bershadsky:1993cx}. To this end an understanding of the geometric quantization of the $tt^*$ equations will be necessary \cite{Witten:1993ed}, see also Ref.~\cite{Li:2011mx} addressing this for the elliptic curve.

\subsection*{Acknowledgments}
I would like to thank Hossein Movasati for collaboration and many discussions of his work which have inspired this project. I would also like to thank Emanuel Scheidegger, Shing-Tung Yau and Jie Zhou for discussions and collaborations on related projects. I am very grateful to Sergio Cecotti, Michele Del Zotto, Thomas Dumitrescu, Sarah Harrison, Shinobu Hosono, Ilarion Melnikov and Johannes Walcher for very helpful comments and discussions. It is furthermore a pleasure to acknowledge the new Center for Mathematical Sciences and Applications at Harvard University for providing an environment which furthers the stimulating interaction between mathematics and physics. This work has been supported by NSF grants PHY-1306313.


\appendix

\section{Vector fields on $\mathsf{T}$}
We give the explicit entries of some matrices which appear in the computation of the Lie algebras for lattice polarized K3 manifolds and for CY threefolds.

\subsection{Lattice polarized K3}
For the matrix $\mathsf{B}$ in
\begin{equation}
\vec{\Omega}_t = \mathsf{B}\, \vec{\Omega}_z\,.
\end{equation}
we have
\begin{equation}
\mathsf{B}^{-1}= \left( \begin{array}{ccc}
g_0 &0&0\\
-L_b & \delta_b^c &0 \\
g_0^{-1} (\frac{1}{2} C^{de} L_d L_e- g_0^2\, \mathcal{E}-g_0 \mathcal{E}^d L_d) &  g_0^{-1}(-C^{cd}L_d + g_0\, \mathcal{E}^c) & g_0^{-1}
\end{array}\right) \,,
\end{equation}
with $\mathcal{E}=(C^{-1})^{*b} k_{*b}$ and $\mathcal{E}^b= C^{*c} s_{*c}^b+ C^{bc} k_c$,

We further find:
\begin{equation}
\frac{\partial}{\partial g_0} \mathsf{B} = \left( \begin{array}{ccc}
-g_0^{-2} &0&0\\
-g_0^{-2}L_a & 0 &0 \\
-g_0^{-2}(\frac{1}{2} C^{de} L_d L_e)+ \mathcal{E}&  -  \mathcal{E}^b & 1
\end{array}\right) \,,   
\end{equation}

and
 
 \begin{equation}
 M_{g_0} = g_0^{-1} \left( \begin{array}{ccc}
-1 &0&0\\
-L_a & 0 &0 \\
0&  -  C^{cd} L_d & 1
\end{array}\right) \,,  
 \end{equation}

 we furthermore compute:

\begin{equation}
\frac{\partial}{\partial K_m} \mathsf{B} = \left( \begin{array}{ccc}
0 &0&0\\
g_0^{-1} \delta_a^m & 0 &0 \\
g_0^{-1}( C^{me} L_e)&  C^{bm} & 0
\end{array}\right) \,,   
\end{equation}

\begin{equation}
M_{K_m} = \left( \begin{array}{ccc}
0 &0&0\\
 \delta_a^m & 0 &0 \\
0&  C^{cm} & 0
\end{array}\right) \,. 
\end{equation}

\subsection{Threefolds}
\begin{equation}
\label{muradalim}
\Am = \left( \begin{array}{cccc} 
g_0^{-1} &0&0&0\\
g_0^{-1} L_a & (g^{-1})^i_a &0&0\\
- g_0^{-1} \widehat{T}^a & (g^{-1})^{i}_d \widehat{T}^{ad} & g_i^a &0\\
g_0^{-1} \left( 2T +\widehat{T}^d L_d\right) -g_0 \mathcal{H} & -(g^{-1})^i_d (T^d + \widehat{T}^{de}L_e) -g_0 \mathcal{H}^{i}&-g_i^e L_e &g_0
\end{array}
 \right)\,,
\end{equation}
where $a$ is an index for the rows and $i$ for the columns and where:
\begin{eqnarray}
\widehat{T}^a &=& T^a -g_0\, g_i^d \mathcal{E}^m\,, \\
\widehat{T}^{ab} &=& T^{ab} - g_m^a \,g_n^b\, \mathcal{E}^{mn}\,,\\
\mathcal{H} &=& g^*_j (g^{-1})_*^i (\partial_i \mathcal{E}^j + C_{imn}\mathcal{E}^{lj}\mathcal{E}^{m}-h_i^j)\\
\mathcal{H}^{i} &=& g_m^* (g^{-1})^n_* (-\partial_n \mathcal{E}^{im} -C_{nlk}\mathcal{E}^{il} \mathcal{E}^{km} + \delta_n^i \mathcal{E}^m + h_n^{im})\\
\mathcal{E}^{ik} &=&(C_{\sharp}^{-1})^{ij} s_{\sharp j}^k\,,\\
\mathcal{E}^{i}&=& (C_{\sharp}^{-1})^{ij} k_{\sharp j}\,.
\end{eqnarray}

We find
\begin{eqnarray}
\label{18.sepetember.2014}
M_{{T^{ab}}}&=& \left(
\begin{array}{cccc}
 0 & 0 & 0 & 0 \\
 0 & 0 & 0 & 0 \\
 -\delta^{i}_{a} L_{b} & \frac{1}{2} (\delta^{i}_{a} \delta^{j}_{b} + \delta^{i}_{b} \delta^{j}_{a} ) & 0 & 0 \\
 L_{a} L_{b} & -\delta^{j}_{b} L_{a} & 0 & 0 \\
\end{array}
\right) \,,\\ \nonumber
M_{{T^{a}}} &=& \left(
\begin{array}{cccc}
 0 & 0 & 0 & 0 \\
 0 & 0 & 0 & 0 \\
 -\delta^{i}_{a} & 0 & 0 & 0 \\
 2 L_{a} & -\delta_{a}^{j} & 0 & 0 \\
\end{array}
\right)\,,\\ \nonumber
M_{T}&=&\left(
\begin{array}{cccc}
 0 & 0 & 0 & 0 \\
 0 & 0 & 0 & 0 \\
 0 & 0 & 0 & 0 \\
 2 & 0 & 0 & 0 \\
\end{array}
\right)\,,\\ \nonumber
M_{{L_{a}}}&=& \left(
\begin{array}{cccc}
 0 & 0 & 0 & 0 \\
 \delta^{a}_{i} & 0 & 0 & 0 \\
 0 & 0 & 0 & 0 \\
 0 & 0 & -\delta^{a}_{j} & 0 \\
\end{array}
\right)\,, \\ \nonumber
M_{{g_0}} &=& \left( \begin{array}{cccc} -g_0^{-1}&0&0&0 \\
 -g_0^{-1} L_{i}&0&0&0\\
 g_0^{-1} T^{i}&0&0&0\\
 -2 g_0^{-1} (2 T_0 +T^{d} L_{d})&g_0^{-1}T^{j}&g_0^{-1} L_{j}&g_0^{-1}
 \end{array}\right) \,.
 \end{eqnarray}

\begin{equation}
M_{{g_{m}^{a}}} =   \left( \begin{array}{cccc} 0&0&0&0 \\
g^{m}_i L_{a} & -\delta^{j}_{a} g^{m}_i &0&0\\
T^{id} g^m_d L_a+\delta_a^i g^m_d (T^{d}+T^{de}L_e) & -\delta^{i}_{a}\,T^{jd} g^m_d- \delta^j_a T^{id} g^m_d &  \delta_a^i (g^{-1})^{m}_j  &0 \\
-2  L_a g^m_d(T^d+T^{de} L_e)  &  \delta^j_a g^m_d (T^d+T^{de}L_e) +g^m_d T^{dj}  L_a & -(g^{-1})^m_j L_a & 0 
\end{array}\right)\,.\\
\end{equation}

\section{Quasi modular forms and differential rings}\label{modularappendix}

We give the expressions of the modular objects that appear in this work, more details can be found in Ref.~\cite{Alim:2013eja} and references therein.
We define (in the literature the choice for $q$ is a matter of convention, in our paper we shall take $q=\exp{2\pi i \tau}$)

\begin{equation}
\vartheta \left[\!\!\! \begin{array}{c}a\\ b\end{array}\!\!\!\right](z,\tau)=\sum_{n\in \mathbbm{Z}}  q^{{1\over 2}(n+a)^2} e^{2\pi i (n+a)(z+b)} \,.
\end{equation}

The following labels are given to the theta functions:
\begin{align}
\theta_1(z,\tau)&=\vartheta \left[\!\!\! \begin{array}{c}1/2\\ 1/2\end{array}\!\!\!\right](u,\tau)=\sum_{n\in \mathbbm{Z}+{1\over 2}}  (-1)^{n}q^{{1\over 2}n^2}e^{2\pi i n z}\,,\\
\theta_2(z,\tau)&=\vartheta \left[\!\!\! \begin{array}{c}1/2\\ 0\end{array}\!\!\!\right](u,\tau)=\sum_{n\in \mathbbm{Z}+{1\over 2}}  q^{{1\over 2}n^2}e^{2\pi i n z}\,,\\
\theta_3(z,\tau)&=\vartheta \left[\!\!\! \begin{array}{c}\,\,\,0\,\,\,\\ \,\,\,0\,\,\,\end{array}\!\!\!\right](u,\tau)=\sum_{n\in \mathbbm{Z}}  q^{{1\over 2}n^2}e^{2\pi i n z}\,,\\
\theta_4(z,\tau)&=\vartheta \left[\!\!\! \begin{array}{c}0\\ 1/2\end{array}\!\!\!\right](u,\tau)=\sum_{n\in \mathbbm{Z}} (-1)^n q^{{1\over 2}n^2}e^{2\pi in z}\,.
\end{align}

We also define the following $\theta$--constants:
\begin{equation}
\theta_{2}(\tau)=\theta_2(0,\tau),\quad \theta_{3}(\tau)=\theta_3(0,\tau),\quad \theta_{4}(\tau)=\theta_2(0,\tau)\,.
\end{equation}
The $\eta$--function is defined by
\begin{equation}
\eta(\tau)=q^{\frac{1}{24}}\prod_{n=1}^\infty(1-q^n)\,.
\end{equation}
It transforms according to
\begin{equation}\label{etatrafo}
\eta(\tau+1)=e^{\frac{i\pi}{12}}\eta(\tau),\qquad \eta\left(-\frac{1}{\tau}\right)=\sqrt{\frac{\tau}{i}}\,\eta(\tau)\,.
\end{equation}
The Eisenstein series are defined by
\begin{equation}\label{eisensteinseries}
E_k(\tau)=1-\frac{2k}{B_k}\sum_{n=1}^\infty\frac{n^{k-1}q^n}{1-q^n},
\end{equation}
where $B_k$ denotes the $k$-th Bernoulli number. $E_k$ is a modular form of weight $k$ for $k>2$ and even. The discriminant form and the $j$
invariant are given by
\begin{align}
   \Delta(\tau) &= \frac{1}{1728}\left({E_4}(\tau)^3-{E_6}(\tau)^2\right) = \eta(\tau)^{24},\\
    j(\tau)& = 1728{E_{4}(\tau)^{3}\over E_{4}(\tau)^3-{E_6}(\tau)^2 }\,.
\end{align}

\subsection{Differential ring}
The modular forms obey the following differential equations:
\begin{align}
\partial_{\tau}\log \eta(\tau)&={1\over 24}E_{2}(\tau)\,,\\
\partial_{\tau}\log \sqrt{\textrm{Im}~\tau}|\eta(\tau)|^2&={1\over 24}\widehat{E_{2}}(\tau,\bar{\tau})\, .
\end{align}
where we denote by $\partial_{\tau} :={1\over 2\pi i}{\partial\over \partial \tau}$, $\widehat{E}_2$ is the non-homolorphic modular completion of the quasi modular form $E_2$. $E_2,E_4 $ and $E_6$ satisfy the following differential ring:

\begin{equation}
\begin{aligned}
                  \partial_\tau E_{2}&={1\over 12}(E_{2}^2-E_{4})\,,\\
                  \partial_\tau E_{4}&={1\over 3}(E_{2}E_{4}-E_{6})\,,\\
\partial_\tau E_{6}&={1\over
2}(E_{2}E_{6}-E_{4}^{2})\,.
                          \end{aligned}
                          \end{equation}

\subsection{Congruence subgroups}
The following the genus zero congruence subgroups called Hecke subgroups of $\Gamma(1)=PSL(2,\mathbb{Z})=SL(2,\mathbb{Z})
/\{\pm I\}$
\begin{equation}
\Gamma_{0}(N)=\left\{
\left.
\begin{pmatrix}
a & b  \\
c & d
\end{pmatrix}
\right\vert\, c\equiv 0\,~ \textrm{mod} \,~ N\right\}< \Gamma(1)
\end{equation}

For these subgroups we introduce three modular forms $A,B,C$ of weight 1, which are given by:
\begin{equation}
    \hspace{2.5em} \begin{array}{c|ccc}\renewcommand{\arraystretch}{0.5}
        N&A&B&C\\[.2ex]
1^{*}&E_{4}(\tau)^{1\over 4}&({E_{4}(\tau)^{3\over 2}+E_{6}(\tau)\over
    2})^{1\over 6}&({E_{4}(\tau)^{3\over 2}-E_{6}(\tau)\over 2})^{1\over
    6}\\[.5ex]
2&{(2^{6}\eta(2\tau)^{24}+\eta(\tau)^{24} )^{1\over 4} \over
\eta(\tau)^2\eta(2\tau)^2}&{\eta(\tau)^{4}\over \eta(2\tau)^{2}}&2^{3 \over
2}{\eta(2\tau)^4\over \eta(\tau)^2}\\[1ex]
3&{(3^{3}\eta(3\tau)^{12}+\eta(\tau)^{12} )^{1\over 3} \over
\eta(\tau)\eta(3\tau)}&{\eta(\tau)^{3}\over \eta(3\tau)}&3{\eta(3\tau)^3\over
\eta(\tau)}\\[1ex]
4&{(2^{4}\eta(4\tau)^{8}+\eta(\tau)^{8} )^{1\over 2} \over \eta(2\tau)^2}=
{\eta(2\tau)^{10}\over\eta(\tau)^{4}\eta(4\tau)^{4}}&{\eta(\tau)^{4}\over \eta(2\tau)^2}&2^2{\eta(4\tau)^4\over \eta(2\tau)^2}
\end{array}
\end{equation}

These satisfy by definition
\begin{equation}
A^{r}=B^{r}+C^{r}\,.
\end{equation}

with the following values of $r$:
\begin{equation*}
 \begin{array}{c|ccccc}
N&1^{*}&2&3&4&\\
r&6&4&3&2
\end{array}
\end{equation*}
We introduce the analog of the Eisenstein series $E_{2}$ as a quasi modular form as follows:
\begin{eqnarray}
E=\partial_\tau \log
B^{r}C^{r}\,.
\end{eqnarray}

The differential ring structure becomes:
\begin{equation}
\begin{aligned}
\partial_\tau A&={1\over 2r}A(E+{C^{r}-B^{r}\over A^{r-2}})\,,\\
\partial_\tau B&={1\over 2r}B(E-A^{2})\,,\\
\partial_\tau C&={1\over 2r}C(E+A^{2})\,,\\
\partial_\tau E&={1\over
2r}(E^{2}-A^{4})\,.
  \end{aligned}
\end{equation}


\providecommand{\href}[2]{#2}\begingroup\raggedright\endgroup


\end{document}